\documentclass[11pt]{article}
\usepackage[utf8]{inputenc}
\usepackage[driverfallback=dvipdfmx,colorlinks=true,citecolor=blue,linkcolor=blue,urlcolor=blue]{hyperref}
\usepackage[left=1in, right=1in, top=1in, bottom=1in]{geometry}

\usepackage[english]{babel} 
\usepackage{amsmath,amssymb,mathrsfs}
\usepackage{bm}
\usepackage{mathtools}
\usepackage{cleveref}
\usepackage{bm}
\usepackage{multirow}
\usepackage{array}
\usepackage[dvipsnames]{color}

\usepackage{booktabs}
\usepackage{algorithm}
\usepackage{algpseudocodex}
\usepackage{comment}
\usepackage{subfigure}
\usepackage{tensor}
\usepackage{enumitem}
\usepackage{nicefrac}
\usepackage{cancel}

\usepackage{amsthm}

\newtheorem{lemma}{Lemma}

\newtheorem{problem}{Problem}

\newtheorem{remark}{Remark}
\newcommand{\norm}[1]{\left\lVert#1\right\rVert}
\newcommand{\abs}[1]{\left\lvert#1\right\rvert}

\let\P\relax
\newcommand{\P}[1]{\mathbb{P}\left[#1\right]}
\newcommand{\E}[1]{\mathbb{E}\left[#1\right]}
\newcommand{\cov}[1]{\mathrm{Cov}\left[#1 \right]}
\newcommand{\R}{\mathbb{R}}

\newcommand{\diff}{\mathrm{d}}

\DeclarePairedDelimiterX{\infdivx}[2]{(}{)}{#1\;\delimsize\|\;#2}

\newcommand{\lbar}[1]{\mskip.5\thinmuskip\overline{\mskip-.5\thinmuskip {#1} \mskip-.5\thinmuskip}\mskip.5\thinmuskip} % long bar

\crefname{align}{Eq.}{Eqs.}
\crefname{appendix}{Appendix}{Appendices}
\crefname{equation}{Eq.}{Eqs.}
\crefname{figure}{Fig.}{Figs.}
\crefname{table}{Table}{Tables}
\crefname{theorem}{Theorem}{Theorems}
\crefname{definition}{Definition}{Definitions}
\crefname{lemma}{Lemma}{Lemmas}
\crefname{remark}{Remark}{Remarks}
\crefname{assumption}{Assumption}{Assumptions}
\crefname{proof}{Proof}{Proofs}
\crefname{problem}{Problem}{Problems}
\crefname{proposition}{Proposition}{Propositions}
\crefname{remark}{Remark}{Remarks}
\crefname{section}{Section}{Sections}
\crefname{algorithm}{Algorithm}{Algorithms}

\algtext*{EndProcedure}% Remove "end while" text

\title{Convex Block-Cholesky Approach to Risk-Constrained Low-thrust Trajectory Design under Operational Uncertainty}

\author{Kenshiro Oguri%
\footnote{Assistant Professor, School of Aeronautics and Astronautics, Purdue University, West Lafayette, Indiana, 47907}\ \ 
and
Gregory Lantoine%
\footnote{Group Supervisor, Outer Planet Mission Design Group, Mission Design and Navigation Section, Jet Propulsion Laboratory, California Institute of Technology, Pasadena, California 91109.
\newline An earlier version of this paper was presented as paper 22-708 at the 2022 AAS/AIAA Astrodynamics Specialist Conference, Charlotte, NC, August 2022.}%
}
\date{}

\begin{document}

\maketitle

\begin{abstract}
Designing robust trajectories under uncertainties is an emerging technology that may represent a key paradigm shift in space mission design.
As we pursue more ambitious scientific goals (e.g., multi-moon tours, missions with extensive components of autonomy), it becomes more crucial that missions are designed with navigation (Nav) processes in mind.
The effect of Nav processes is statistical by nature, as they consist of orbit determination (OD) and flight-path control (FPC).
Thus, this mission design paradigm calls for techniques that appropriately quantify statistical effects of Nav, evaluate associated risks, and design missions that ensure sufficiently low risk while minimizing a statistical performance metric;
a common metric is $\Delta$V99: worst-case (99\%-quantile) $\Delta$V expenditure including statistical FPC efforts.
In response to the need, this paper develops an algorithm for risk-constrained trajectory optimization under operational uncertainties due to initial state dispersion, navigation error, maneuver execution error, and imperfect dynamics modeling.
We formulate it as a nonlinear stochastic optimal control problem and develop a computationally tractable algorithm that combines optimal covariance steering and sequential convex programming (SCP).
Specifically, the proposed algorithm takes a block-Cholesky approach for convex formulation of optimal covariance steering, and leverages a recent SCP algorithm, \texttt{SCvx*}, for reliable numerical convergence.
We apply the developed algorithm to risk-constrained, statistical trajectory optimization for exploration of dwarf planet Ceres with a Mars gravity assist, and demonstrate the robustness of the statistically-optimal trajectory and FPC policies via nonlinear Monte Carlo simulation.
\end{abstract}

%!TEX root = main.tex
\section{Introduction}
\label{sec:intro}

% about trajectory design under uncertainty -- paradigm shift
Designing robust trajectories under uncertainties is an emerging technology that may represent a key paradigm shift in space mission design.
As the complexity of space missions grows to pursue more ambitious scientific goals (e.g., multi-moon tours, missions with extensive components of autonomy), it becomes increasingly crucial that missions are designed with navigation (Nav) processes in mind.
The effect of Nav processes is statistical by nature, as they consist of orbit determination (OD) and flight-path control (FPC);
FPC is a process that calculates trajectory correction maneuvers (TCMs) based on delivered OD solutions.%
\footnote{We are aware that there are different interpretations/definitions of ``navigation'';
it sometimes refers to only estimation process (i.e., OD) while FPC is considered a ``guidance'' process.
In this paper, we take the definition typically used at JPL, where Nav encompasses OD and FPC.}
Any risks (e.g., collision with a planetary body during a close flyby) under these statistical effects must be appropriately quantified and mitigated in mission design processes to ensure the mission safety.
Although these risks have been historically mitigated in a rather heuristic manner (e.g., heuristic margins to constraints, manually-tuned forced coasting arcs) \cite{Rayman2007,Oh2008}, there has been a growing number of studies on directly incorporating these statistical effects into mission design and thereby coupling the trajectory optimization with Nav processes within the design process.

% existing work on robust trajectory optimization under uncertainty --- discuss their pros and cons
Similar to conventional deterministic problems, existing approaches to statistical trajectory optimization under uncertainties can be categorized into:
indirect methods \cite{Zimmer2010,Jenson2021a,Oguri2022c},
direct methods \cite{Greco2022,Ridderhof2020c,Oguri2022f,Benedikter2022a,Varghese2025,Kumagai2025}
and
differential dynamic programming (DDP) \cite{Ozaki2019,Yuan2024}.
This list of papers are certainly not exhaustive, but it already highlights the growing interest on this topic in the space mission design community.
See \cite{Kumagai2025} and its earlier version \cite{Kumagai2024b} for further discussion that compares these different methods.
Key enablers of these techniques lie in the recent advances in stochastic optimal control and optimization communities, particularly in
\textit{optimal covariance steering} \cite{Okamoto2018a,Chen2018c,Okamoto2019a,Ridderhof2020,Aleti2023,Liu2024,Pilipovsky2024a}
and
\textit{computational control} based on convex programming \cite{Acikmese2007,Ackmese2011,Mao2016,Bonalli2022a,Oguri2023b},
and their application to aerospace problems \cite{Lu2013,Szmuk2020,Oguri2021e,Oguri2024,Ra2024,Sagliano2024}.
Unlike traditional optimal control, optimal covariance steering consider controlling the mean and covariance of the system's state while minimizing a statistical cost, which depends on both the state mean and covariance;
it often also involves statistical constraints known as \textit{chance constraints} (or equivalently, \textit{risk constraints}), which imposes constraints on probability that constraints be satisfied under state uncertainties.
Either explicitly or implicitly, many existing techniques for trajectory optimization under uncertainties build on the concept of optimal covariance steering, solved by indirect methods, direct methods, or DDP.

% more discussion on SOTA of covariance steering and their pros and cons
A particularly important aspect to consider when developing an algorithm for trajectory optimization under uncertainty is, \textit{what specific approach we choose to formulate the optimal covariance steering problem}.
First introduced in 1980s by Hotz and Skelton \cite{HOTZ1987}, the theory of optimal covariance steering has seen rapid advances in the stochastic optimal control community in recent years, including
continuous-time covariance steering \cite{Chen2018c}
and
discrete-time covariance steering with chance constraints \cite{Okamoto2018a}.
Discrete-time covariance steering may find more relevant applications in aerospace, considering the nature of discrete opportunities for control commands and measurements.
It is therefore beneficial to further discuss some variants of discrete-time covariance steering, such as
one with input constraints \cite{Bakolas2018},
one without the need for history feedback \cite{Okamoto2019a},
one with output feedback \cite{Ridderhof2020},
and
output-feedback covariance steering without history feedback \cite{Aleti2023}.
In particular, this \textit{output-feedback} covariance steering holds the key to unlocking the potential of incorporating the OD process (i.e., taking measurements and applying filtering to estimate the spacecraft orbit) into the framework of trajectory optimization under uncertainties.
Notably, the discrete-time covariance steering studies mentioned in this paragraph formulate the problem in convex form by forming a large block-matrix for the system equation and utilizing the Cholesky factor of covariance matrices. 
This class of approaches are called \textit{block-Cholesky formulation} in this paper.

An important caveat of the \textit{block-Cholesky formulation} lies in its computational complexity.
As discussed in \cite{Pilipovsky2024a,Oguri2022f,Kumagai2025}, their computational complexity is roughly proportional to $(n_x N^2)$, where $n_x$ is the size of the state vector and $N$ is the number of discretized trajectory segments.
Hence, recent studies have developed a more computationally-efficient formulation of discrete-time covariance steering \cite{Liu2024} and its output-feedback variant \cite{Pilipovsky2024a}, which have been successfully applied to robust space trajectory optimization under uncertainties \cite{Benedikter2022a,Kumagai2024b,Kumagai2025}.
These approaches propagate the full covariance (rather than its Cholesky factor) and do not rely on the block-matrix formulation, and hence called \textit{full-covariance formulation} in this paper.
The main innovation of these studies lies in its elegant lossless convexification of originally non-convex covariance propagation equations by using the Schur complement and the Karush-Kuhn-Tucker (KKT) conditions.
As a result, the \textit{full-covariance formulation} enables us to formulate the discrete-time covariance steering problem in convex form, with its computational complexity roughly proportional to $N (n_x + n_u )$, where $n_u$ is the size of the control vector.
This implies that the full-covariance formulations is more beneficial for larger trajectory optimization problems (i.e., greater $N$).

% scaling issue + non-convex chance constraints + objective function
Although the full-covariance approach has a key advantage over block-Cholesky approaches in terms of the computational complexity, they also have critical drawbacks that warrant careful considerations for space trajectory optimization.
The drawbacks are:
(1) numerical ill-conditioning,
(2) non-convex chance constraints,
and
(3) inflexible cost functions.
The first point, numerical ill-conditioning, is related to the large scale differences that are inherent to space trajectories, especially in interplanetary trajectories.
To illustrate this aspect, let us consider spacecraft on a trajectory from Mars to a main-belt asteroid with state uncertainty of 10 km in position and 1 m/s in velocity (standard deviation).
While the magnitudes of these position/velocity uncertainty variances are of $\{100\ \mathrm{km}^2, 1\times10^{-6}\ \mathrm{km^2/s^2}\}$, the magnitude of position vector measured from Sun is in the order of $\sim 2{-}4\times 10^8$ km;
this implies that the numerical solver would need to deal with a problem of condition number as large as $\sim 1\times10^{14} $ and may not be as reliable as desired.%
\footnote{If we normalize the problem to make both 1 AU and 1 $\mu_{\mathrm{sun}}$ equal to unity, the condition number might become even worse. In the case of the example above, the variances of position and velocity become approximately $\{4.4\times10^{-15}, 1.1\times10^{-17}\}$ against 2-3 AU.}
In fact, this type of numerical ill-conditioning is the main motivation for square-root formulations of OD filters (e.g., square-root unscented Kalman filter \cite{Geeraert2017} and square-root information filter \cite{Tapley2004}), which can recover up to a half of the digits of precision that would have been lost otherwise.

The second drawback, non-convex chance constraints, is due to the use of full covariance matrices as decision variables.
As demonstrated in existing studies on chance-constrained control \cite{Okamoto2018a,Okamoto2019a,Ridderhof2020,Aleti2023,Oguri2024,Ra2024}, we can formulate many chance constraints in convex form \textit{if} the Cholesky factor of covariance matrices are decision variables.
On the other hand, full-covariance approaches directly handle covariance matrices, which is a nonlinear function of the Cholesky factor, hence making those chance constraints non-convex, including control magnitude chance constraints and half-plane state chance constraints, to name a few.
This fact requires us to approximate chance constraints at each iteration via either linearization \cite{Benedikter2022a,Kumagai2024b,Kumagai2025} or convex-concave procedure \cite{Pilipovsky2024a}.

Lastly, the issue of inflexible cost functions arises from the lossless convexification process in the formulation of full-covariance approaches.
The lossless property, which is crucial for convexifying the full covariance propagation equation, relies on the KKT conditions, in particular their complementary slackness condition.
To meet the KKT conditions, \cite{Liu2024,Pilipovsky2024a} assume a specific form of objective function, and it is shown \cite{Rapakoulias2023} that their full-covariance approaches cannot consider the $\Delta$V99 cost metric (99\%-quantile $\Delta$V, representing worst-case $\Delta$V under statistical FPC), which is a typical cost in space mission design.
\cite{Kumagai2025} proposes a middle-ground solution that relaxes the inflexibility and enables approximate formulation of the $\Delta$V99 cost.

% what we propose -- what limitations we address: 
Given these observations, this paper develops an approach based on block-Cholesky formulation for robust trajectory optimization under uncertainties.\footnote{A recent paper of the first author of this paper develops an algorithm based on a full-covariance approach \cite{Kumagai2025}, if the readers are interested in the other formulation too.}
Main contributions of the present paper are threefold.
First, this paper leverages a recent formulation of block-Cholesky output-feedback covariance steering \cite{Oguri2024} and a sequential convex programming (SCP) algorithm \texttt{SCvx*} \cite{Oguri2023b} to develop an optimization algorithm for risk-constrained low-thrust trajectory design under uncertainties.
The algorithm incorporates various operational uncertainties due to navigation errors, maneuver execution errors, initial state dispersion, and imperfect dynamics modeling.
Second, we extend the developed algorithm to incorporate planetary gravity assists (GAs), which is a must-have option in any modern interplanetary mission design.
In particular, we utilize the patched-conic GA model in this study.
Third, the developed algorithm introduces some improvements in its mathematical formulation, including:
a more general class of augmented Lagrangian functions than introduced in the original \texttt{SCvx*} paper \cite{Oguri2023b}, which can smoothly approximate the $l_1$-penalty function for better convergence;
and
a formulation of three-dimensional control magnitude chance constraints that provide tighter approximation than existing studies \cite{Ridderhof2020c,Benedikter2022a,Zhang2023,Yuan2024} and are rigorously conservative compared to a heuristic approach taken in \cite{Ozaki2019}.
See \cref{f:concept} for a conceptual illustration of the proposed approach.

The rest of the paper is structured as follows.
After \cref{sec:EoM} introduces the system equation, \cref{sec:SOCformulation} details the stochastic optimal control formulation that incorporates operational uncertainties and statistical risk constraints.
\cref{sec:tractableFormulation} presents a tractable convex approximation of the stochastic optimal control problem based on a block-Cholesky approach, which simultaneously optimizes the reference trajectory and FPC policies under the Kalman filter process and linear FPC policies.
\cref{sec:solutionMethod} develops a SCP algorithm that leverages the algorithmic framework of \texttt{SCvx*} and iteratively solves the convex subproblem to march toward an optimal solution.
\cref{sec:examples} applies the developed algorithm to robust trajectory optimization for exploration of dwarf planet Ceres with a Mars gravity assist, and demonstrates the robustness of the obtained statistically-optimal trajectory and associated FPC policies via nonlinear Monte Carlo simulation.

%!TEX root = main.tex
\section{Equations of Motion}
\label{sec:EoM}

This paper considers low-thrust spacecraft orbital motions with a simplified low-thrust model that ignores the mass decrease due to thrusting.
This assumption is in line with most of the existing approaches to robust space trajectory optimization under uncertainty, including \cite{Ridderhof2020c,Ozaki2019,Benedikter2022a,Jenson2021a}.
While some studies, such as \cite{Oguri2022f,Benedikter2023,Varghese2025}, consider the mass flow equation and associated uncertainty, they also point out that the resulting distribution of the mass uncertainty becomes highly non-Gaussian and in fact take the form of generalized chi-squared distribution.
It is well-known that the generalized chi-squared distribution does not have closed-form expression for the probability density function, making rigorous characterization of mass uncertainty less straightforward.
Developing an approach to properly incorporating the effect of mass uncertainty is an important direction of future work.

Thus, letting $\bm{x}\in\R^{n_x} $ denote the orbital state (e.g., position/velocity, osculating orbital elements, etc) and $\bm{u}\in\R^3$ be the control acceleration due to the low-thrust engine, we can express the equations of motion in a generic form as follows:
\begin{align} \begin{aligned}
\dot{\bm{x}} =& \bm{f}(\bm{x}, \bm{u}, t) = 
\bm{f}_0(\bm{x},t) + F(\bm{x})\bm{u} 
\label{eq:genericEoM}
\end{aligned} \end{align}
where $\bm{f}_0(\cdot,\cdot):\R^{n_x}\times\R\mapsto\R^{n_x} $ represents natural orbital dynamics;
$F(\cdot): \R^{n_x} \mapsto \R^{n_x\times 3} $ returns a matrix that maps the acceleration to the rate of the state change.
\cref{eq:genericEoM} encompasses equations of motion describing low-thrust orbital dynamics in various coordinate systems.

Specifically, this study uses the Cartesian coordinate system to express the system state and considers perturbed two-body dynamics.
Then, the orbital state and the dynamics functions are given by:
\begin{align}
\begin{aligned}
\bm{x} = 
\begin{bmatrix}
\bm{r} \\ \bm{v}
\end{bmatrix}
,\quad
\bm{f}_0 (\bm{x},t) = 
\begin{bmatrix}
\bm{v} \\ -\frac{\mu}{\norm{\bm{r}}_2^3}\bm{r} + \bm{a}_{\mathrm{pert}}(\bm{x},t)
\end{bmatrix},\quad
F = 
\begin{bmatrix}
0_{3\times 3} \\ I_3
\end{bmatrix}
\end{aligned}\end{align}
where $\bm{r}\in\R^3$ and $\bm{v}\in\R^3$ are the spacecraft position and velocity in Cartesian coordinates, respectively;
$\mu$ is the gravitational parameter of the central body;
$\bm{a}_{\mathrm{pert}}(\bm{x},t)$ is the perturbing acceleration.

%!TEX root = main.tex
\section{Risk-Constrained Nonlinear Stochastic Optimal Control Problem}
\label{sec:SOCformulation}

\begin{figure}[tb]
\centering 
\includegraphics[width=\linewidth]{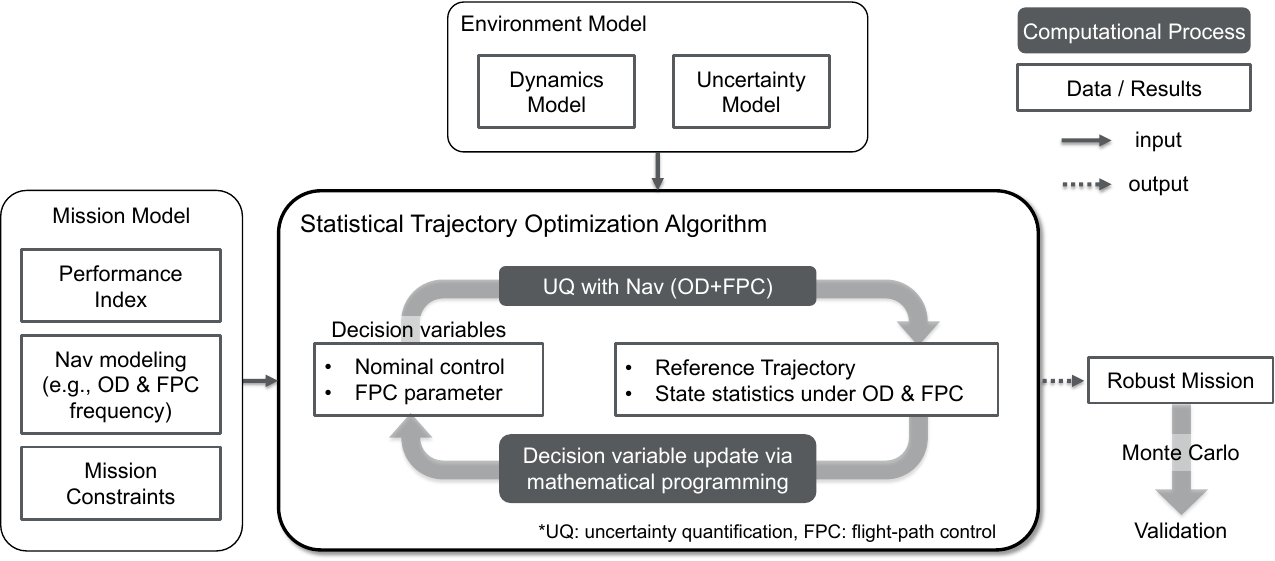}
	\caption{\label{f:concept} Fundamental idea of statistical trajectory optimization under uncertainty}
\end{figure}

\cref{f:concept} presents the fundamental idea behind the proposed formulation for risk-constrained trajectory optimization under uncertainties.
The formulation directly incorporates operational uncertainty models (e.g., measurement function, measurement noise, OD frequency) as well as statistical risk constraints (e.g., safety with 99.9\% confidence) while minimizing the typical statistical cost, $\Delta$V99.
Considering other statistical costs is straightforward due to the flexibility of block-Cholesky approaches as opposed to full-covariance approaches.
The developed algorithm then solves for statistically optimal low-thrust trajectories under the modeled uncertainties by utilizing the mean and covariance information derived from the Kalman filter process, and simultaneously optimizing the reference trajectory and FPC policies via SCP.

This section formulates the problem of risk-constrained low-thrust trajectory optimization under uncertainty in a form of chance-constrained optimal control.
Our formulation incorporates the influence of major uncertain errors in space missions (initial state dispersion, maneuver execution error, and OD error) as well as the corrective effect of FPC efforts.

\subsection{Nonlinear Stochastic System}

\subsubsection{Initial State Dispersion}

Any space trajectories begin with initial state dispersion, such as due to launch dispersion and orbit insertion errors.
We model such initial state dispersion via a Gaussian distribution as:
\begin{align}
\bm{x}_0 \sim \mathcal{N}(\lbar{\bm{x}}_0, P_0)
\end{align}

\subsubsection{Maneuver Execution Error}

Any maneuvers involve some errors that are uncertain by nature, due to thruster misalignment, over/under burn, among other causes.
A common approach to modeling such errors is the Gates model \cite{Gates1963}, which captures the error that is proportional to the applied thrust commands as well as that is fixed irrespective of the commands.
Using the Gates model, for a given commanded control $\bm{u} $, a maneuver execution error, denoted by $\widetilde{\bm{u}}$, is given as \cite{Oguri2021e}:
\begin{align}\begin{aligned}
&
\widetilde{\bm{u}} = G_{\mathrm{exe}}(\bm{u}) \bm{w}_{\mathrm{exe}}
,\quad
G_{\mathrm{exe}}(\bm{u}) = T(\bm{u}) P_\mathrm{gates}^{1/2}(\bm{u})
,
\quad  
\bm{w}_{\mathrm{exe}} \sim \mathcal{N}(0, I_3),
\label{eq:executionError}
\end{aligned}\end{align}
where
$\bm{w}_{\mathrm{exe}} $ are independent and identically distributed (i.i.d.) standard Gaussian vectors, i.e., $\E{\bm{w}_{\mathrm{exe}}}=0 $ and $\E{\bm{w}_{\mathrm{exe}}(t_i)\bm{w}_{\mathrm{exe}}(t_j)^\top} = \delta_{i,j}I  $, where $\delta_{i,j}$ is the Kronecker delta.
$P_\mathrm{gates}(\cdot)$ denotes the execution error covariance in the local coordinates that are defined with respect to the thrust direction while $T(\cdot)$ rotates the direction to the coordinates we propagate our system, given by
\begin{align}\begin{aligned}
&
P_\mathrm{gates}(\cdot) =
\mathrm{diag}(\sigma_p^2, \sigma_p^2,  \sigma_m^2 )
,\quad 
T(\cdot) = 
\begin{bmatrix}
\hat{\bm{S}} & \hat{\bm{E}} & \hat{\bm{Z}}
\end{bmatrix}
,\quad
\hat{\bm{Z}} = \frac{\bm{u}}{\norm{\bm{u}}_2} 
,\quad 
\hat{\bm{E}} = \frac{[0,0,1]^\top \times \hat{\bm{Z}}}{\norm{[0,0,1]^\top \times \hat{\bm{Z}}}_2}
,\ 
\hat{\bm{S}} = \hat{\bm{E}} \times \hat{\bm{Z}}
\label{eq:GatesModel}
\end{aligned}\end{align}
Here, 
$\sigma_p^2 = \sigma_3^2 + \sigma_4^2 \norm{\bm{u}}_2^2 $ and $\sigma_m^2 = \sigma_1^2 + \sigma_2^2 \norm{\bm{u}}_2^2$, where $\{\sigma_1, \sigma_2\} $ are \{fixed, proportional\} magnitude errors while $\{\sigma_3, \sigma_4\}$ are \{fixed, proportional\} pointing errors.
It is clear that the terms $\sigma_2^2 \norm{\bm{u}}_2^2 $ and $\sigma_4^2 \norm{\bm{u}}_2^2 $ take greater values for larger control command, and hence represent proportional errors.

\subsubsection{Stochastic Differential Equations}
To incorporate uncertain errors and FPCs in our mission design process, one of the most straightforward approaches is to express our system via nonlinear stochastic differential equations (SDEs).
Based on the equations of motion \cref{eq:genericEoM}, our stochastic system can be formally expressed via a set of nonlinear SDEs as:
\begin{align}\begin{aligned}
\diff\bm{x} =& 
[\bm{f}(\bm{x}, \bm{u}, t) + F(\bm{x}) \widetilde{\bm{u}}] \diff t + G(\bm{x})\diff\bm{w}(t),
\label{eq:SDE}
\end{aligned}\end{align}
where $G(\cdot): \R^{n_x} \mapsto \R^{n_x\times n_w} $ represents the intensity of disturbances, mapping $\diff\bm{w}(t)$ to the random variation of the state;
$\diff\bm{w}(t):\R\mapsto \R^{n_w} $ is a standard Brownian motion vector, i.e., $\E{\diff\bm{w}}=0 $ and $\E{\diff\bm{w}(t)\diff\bm{w}^\top(t)} = I\diff t  $.
The form $G(\cdot)\diff\bm{w}$ encompasses various types of disturbances, including unmodeled external perturbations (e.g., solar radiation pressure, momentum de-saturation operations, etc).
The term $F(\bm{x}) \widetilde{\bm{u}}$ captures the effect of maneuver execution errors to the variation of state uncertainty.

\subsection{Nonlinear Stochastic System with Navigation Uncertainty}

Let us then model the Nav processes (OD and FPC).
We model the FPC policy such that calculates trajectory corrections based on imperfect state knowledge given by orbit determination (OD) processes.

\subsubsection{Orbit Determination Model}
We model the OD process as a sequence of filtering processes with discrete-time observations.
Assuming that a set of spacecraft tracking data becomes available at time $t_k$, the observation process is given by:
\begin{align}\begin{aligned}
\bm{y}_k =& \bm{f}_{\mathrm{obs}}(\bm{x}_k, t_k) + G_{\mathrm{obs}}(\bm{x}_k) \bm{w}_{\mathrm{obs},k},
\quad k=0,1,...,N
\label{eq:SOCobservations}
\end{aligned}\end{align}
where $\bm{y}_k\in\R^{n_y} $ are the (collection of) measurements;
$\bm{f}_{\mathrm{obs}}(\cdot,\cdot): \R^{n_x}\times\R\mapsto \R^{n_y} $ and $G_{\mathrm{obs}}(\cdot): \R^{n_x}\mapsto \R^{n_y\times n_y} $ model the observation process.
$\bm{w}_{\mathrm{obs},k}\in\R^{n_y} $ is measurement noise, modeled as a sequence of i.i.d.~standard Gaussian random vectors.
% $k(=0,1,...,N)$ denotes the discrete time steps where observations $\bm{y}_k$ are available.

At time $t_k$, the OD solution, $\hat{\bm{x}}_k$, is obtained by a filtering process with the past observations $\bm{y}_i (i=0,1,...,k)$ as follows:
\begin{align}\begin{aligned}
\hat{\bm{x}}_k = \mathcal{F}_k(\hat{\bm{x}}_0^{-}, \bm{y}_i: i=0,1,...,k),
\label{eq:SOCfiltering}
\end{aligned}\end{align}
where the ``$-$'' superscript indicates a quantity right before the measurement update;
$\hat{\bm{x}}_0^{-} $ denotes the initial state estimate.
% The estimation error is denoted by $\tilde{\bm{x}}_k = \bm{x}_k - \hat{\bm{x}}_k $.
Filtering typically utilizes the innovation process $\widetilde{\bm{y}}_k^{-}$, defined as:
\begin{align}\begin{aligned}
\widetilde{\bm{y}}_k^{-} = \bm{y}_k - \bm{f}_{\mathrm{obs}}(\hat{\bm{x}}_k^-)
\label{eq:innovationProcess}.
\end{aligned}\end{align}

\subsubsection{Flight-path Control Model}
Flight-path controls are a collection of efforts that drive the orbital state deviation back to the nominal trajectory, planed based on imperfect OD solutions.
Hence, this study models FPCs over the course of a trajectory by a sequence of feedback policies that compute necessary control modifications based on the current state estimate given by \cref{eq:SOCfiltering}.
The feedback policies are thus expressed as:
\begin{align}\begin{aligned}
\bm{u}_k = \pi_k(\hat{\bm{x}}_k, \Omega_k),
\label{eq:SOCpolicy}
\end{aligned}\end{align}
where $\Omega_k$ denotes a set of parameters used to compute the corrective controls $\bm{u}_k$ according to the policy $\pi_k$.
% $\bm{x}^{*}$ and $\bm{u}^{*}$ denote the nominal trajectories of the state and control, respectively.

\subsection{Original Problem}

\subsubsection{Statistical Cost Metric}

A common cost function in space trajectory optimization is the total control effort.
However, we may not directly minimize the control effort in stochastic setting because it is not well-defined due to the stochasticity of $\bm{u}$ under the feedback policy \cref{eq:SOCpolicy}.
Thus, this study considers minimizing the $p$-quantile of the closed-loop control effort, i.e.,
\begin{align} \begin{aligned}
J = \int_{t_0}^{t_f} Q_{X\sim\norm{\bm{u}}_2}(p)
\label{eq:SOCobjective99fuel}
\end{aligned} \end{align}
where 
$Q_{X}(p)$ is the quantile function of a random variable $X$ evaluated at probability $p$, formally defined as:
\begin{align}
Q_{X}(p) = \min\{x\in\R \mid \P{X \leq x} \geq p\}.
\label{eq:defQuantile}
\end{align}
Note that using $p=0.99$ corresponds to minimizing 99\%-quantile of control cost, i.e., $\Delta$V99.

\subsubsection{Risk Constraints}
The constraints considered in this study include path constraints and terminal constraints.
As our state and control are subject to uncertainty, those constraints are not deterministic anymore and need to be treated stochastically.

In particular, we consider a class of probabilistic constraints called \textit{chance constraints} to formulate our path constraints, as they naturally represent stochastic analogs of deterministic constraints classically considered in mission design.
Intuitively, a chance constraint imposes a bound on the probability that a particular constraint be satisfied (i.e., no collision with an obstacle with a probability greater than $99.9\%$);
see, for instance, \cite{Oguri2021e,Greco2022a,Benedikter2022,Oguri2022c,Kumagai2024b} and references therein, for the use of chance constraints in the context of space trajectory optimization.

Mathematically, path chance constraints imposed on discrete-time state and control are expressed as:
\begin{align} \begin{aligned}
\P{c_j(\bm{x}, \bm{u}, t) \leq 0} \geq 1 - \varepsilon_j,\quad
\label{eq:SOCconstraints}
\end{aligned} \end{align}
where $\P{\cdot}$ denotes the probability operator, $c_j(\cdot): \R^{n_x}\times\R^{3}\times\R \mapsto \R  $ represents the $j$-th nonlinear constraint function, and $\varepsilon_j$ is a risk bound associated with the $j$-th constraint (e.g., $\varepsilon_j=0.001$ for $99.9\%$ confidence).
This imposes the constraint $c_j\leq 0$ to be satisfied with probability greater than or equal to $1-\varepsilon_{j}$.

Similarly, we model the terminal constraints as distributional constraints that ensure that the final state is within a prescribed region.
Specifically, we impose terminal constraints on the final state, $\bm{x}_N$, to achieve the target $\bm{x}_f$ on average with prescribed accuracy represented by the final covariance $P_f$:
\begin{subequations}
\begin{align}
&\E{\bm{x}_N} - \bm{x}_f = 0,
\label{eq:terminalMeanOriginal}
\\
&\cov{\bm{x}_N} \preceq P_f 
\label{eq:terminalCovarianceOriginal}
\end{align}
\label{eq:SOCterminal}
\end{subequations}
where $\E{\cdot}$ and $\cov{\cdot}$ denotes the expectation and covariance operators, respectively.
Mathematically, $\E{\bm{x}} = \int \bm{x}p(\bm{x})\diff \bm{x} $ where $p(\bm{x})$ is a probability density function (pdf) of $\bm{x}$, and $\cov{\bm{x}} = \E{(\bm{x} - \E{\bm{x}})(\bm{x} - \E{\bm{x}})^\top} $.

Thus, letting $\bm{\Theta}$ represent a vector of relevant parameters to optimize (which take different forms depending on specific problems), our original problem can be expressed as in \cref{prob:originalSOC}.
\begin{problem}[Original problem]
\label{prob:originalSOC}
Find $\bm{x}^{*}(t)$ $\pi_k, \forall k $, and $\bm{\Theta}$ that minimize the cost \cref{eq:SOCobjective99fuel} subject to the dynamical constraint \cref{eq:SDE}, path chance constraints \cref{eq:SOCconstraints}, and terminal constraint \cref{eq:SOCterminal} under a sequence of FPCs determined by \cref{eq:SOCpolicy} with the state estimates \cref{eq:SOCfiltering} based on observations \cref{eq:SOCobservations}.
\end{problem}

%!TEX root = main.tex
\section{Tractable Formulation via block-Cholesky Covariance Steering}
\label{sec:tractableFormulation}

In general, \cref{prob:originalSOC} is intractable to solve.
A major bottleneck of the difficulty lies in the process of \textit{uncertainty quantification} (UQ) to evaluate the statistics of the state, cost, and constraints under uncertainty, which is computationally expensive in general, except for some special cases.
In nonlinear trajectory optimization, it is most likely the case that we must iteratively perform such expensive UQ processes to evaluate the functions and their gradients with respect to decision variables.
This requires an efficient and reliable solution method.

Among multiple possible approaches,
this study proposes a solution method based on optimal covariance steering and sequential convex programming (SCP) to solve our problem.
The proposed approach approximately formulates the originally intractable stochastic problem in a convex, deterministic form at each successive iteration;
in particular, we employ a block-Cholesky formulation as discussed in \cref{sec:intro}.
At each iteration, the original problem is approximated as a convex subproblem, which is solved with guaranteed convergence to the global optimum of the approximated problem due to the convex property.
Once the convergence of a convex subproblem is achieved, then the subproblem solution is used as a new reference trajectory for the next iteration.
Leveraging the reliability and efficiency of convex programming, this solution method provides us with a vehicle for overcoming the major bottleneck of numerical solutions to nonlinear stochastic optimal control problems, at the cost of approximate dynamical evolution of the stochastic state deviations from the reference trajectory.
Note that the dynamical feasibility of the reference trajectory is not compromised in our solution method, and the designed reference trajectory satisfies the nonlinear dynamics up to the user-specified tolerance.

Let us first approximately reformulate \cref{prob:originalSOC} in a deterministic, convex form in the following.

\subsection{Linear State Statistics Dynamics}
\label{sec:linearStatistics}

To obtain the convex subproblem, \cref{eq:SDE} is first linearized and discretized assuming fixed-time problems (i.e., $t_0$ and $t_f$ not part of optimization variables).
It is possible to extend the formulation presented in this article to variable-time problems by incorporating the adaptive-mesh mechanism \cite{Kumagai2024c}.
Such an extension is left as exciting future work, as it would allow the FPC policies (which are being optimized) to adjust maneuver epochs based on OD solutions.

\subsubsection{Linear, Discrete-time System}
The initial \textit{a priori} state estimate $\hat{\bm{x}}_0^{-}$ and its error $\widetilde{\bm{x}}_0 ^{-}$ (e.g., state dispersion estimate immediately after deployment from a launcher) are assumed to be Gaussian-distributed, i.e., $\hat{\bm{x}}_0^{-}\sim\mathcal{N}(\lbar{\bm{x}}_0, \hat{P}_0^{-}) $ and $\widetilde{\bm{x}}_0 ^{-}\sim\mathcal{N}(0, \widetilde{P}_0^-) $, where $\lbar{(\cdot)}$ indicates the mean of a random variable.
It implies that the initial state $\bm{x}_0^{-} $ is distributed as $\bm{x}_0 = \hat{\bm{x}}_0^{-} + \widetilde{\bm{x}}_0 ^{-} \sim \mathcal{N}(\lbar{\bm{x}}_0, \hat{P}_0^{-} + \widetilde{P}_0^-) $.
Linearizing the system \cref{eq:SDE} and discretizing the time into intervals $t_0 < t_1 < t_2 < ... < t_N = t_f$ with zeroth-order-hold (ZoH) control, i.e., $\bm{u}(t) = \bm{u}_k,\ \forall t\in[t_k, t_{k+1}) $, yields \cite{Oguri2022d,Oguri2024}:
\begin{align}\begin{aligned}
\bm{x}_{k+1} = A_k\bm{x}_k + B_k \bm{u}_k + \bm{c}_k + G_{\mathrm{exe}, k} \bm{w}_{\mathrm{exe}, k}  + G_{k} \bm{w}_{k} 
,\quad
k=0,1,...,N-1
\label{eq:SDElinDiscrete}
\end{aligned}\end{align}
where $A_k, B_k, \bm{c}_k$, and $G_{k}$ represent linear system matrices that are obtained by linearizing about the reference state $\bm{x}^{\mathrm{ref}}(t) $ and control $\bm{u}^{\mathrm{ref}}(t) $ and discretizing the system in time (see Ref.~\cite{Oguri2024} for the definition);
$\bm{w}_{k} $ is an i.i.d.~standard Gaussian random vector.

\subsubsection{Filtered State Dynamics}

Since the state is initially Gaussian-distributed and obeys the linear dynamics in a convex subproblem, the optimal, unbiased state estimate can be obtained by the Kalman filter.
Thus, the filtering process \cref{eq:SOCfiltering} is given as:
\begin{align}\begin{aligned}
&
\begin{cases}
\hat{\bm{x}}_{k}^{-} = A_{k-1}\hat{\bm{x}}_{k-1} + B_{k-1} \bm{u}_{k-1} + \bm{c}_{k-1}\\
\widetilde{P}_{k}^{-} = A_{k-1}\widetilde{P}_{k-1}A_{k-1}^\top + G_{\mathrm{exe}, k-1}G_{\mathrm{exe}, k-1}^\top + G_{k-1}G_{k-1}^\top 
\end{cases}
\quad && (\text{time update})
\\
&
\begin{cases}
\hat{\bm{x}}_{k} = \hat{\bm{x}}_{k}^{-} + L_k(\bm{y}_k - C_k\hat{\bm{x}}_k^-)\\
\widetilde{P}_{k} 	= (I - L_k C_k) \widetilde{P}_{k}^{-} (I - L_k C_k)^\top + L_kD_k D_k^\top L_k^\top,
\end{cases}
&& (\text{measurement update})
\label{eq:estProcessLin}
\end{aligned}\end{align}
where $\widetilde{P}_{k}$ denotes the estimate error covariance;
$L_k$ is the Kalman gain;
$C_k$ and $D_k$ are the linearized measurement equation about the reference, given by
\begin{align}\begin{aligned}
\widetilde{P}_{k}
\triangleq
\cov{\widetilde{\bm{x}}_k}
,\quad
L_k \triangleq \widetilde{P}_k^{-}C_k^\top(C_k \widetilde{P}_k^{-} C_k^\top + D_k D_k^\top)^{-1}
,\quad
C_k \triangleq \frac{\partial \bm{f}_{\mathrm{obs}}}{\partial \bm{x}}(\bm{x}_k^{\mathrm{ref}})
,\quad
D_k = G_{\mathrm{obs}} (\bm{x}_k^{\mathrm{ref}})
\end{aligned}\end{align}
As classically known in estimation theory, $\widetilde{P}_{k}$ and $L_k$ can be computed \textit{a priori} for linear systems.
Combining the two equations in \cref{eq:estProcessLin}, the filtered state process under stochastic error and Kalman filtering is written as:
\begin{align}\begin{aligned}
\hat{\bm{x}}_{k+1} = A_k\hat{\bm{x}}_k + B_k \bm{u}_k + \bm{c}_k + L_{k+1} \widetilde{\bm{y}}_{k+1}^{-} 
\label{eq:linEstProcess}
\end{aligned}\end{align}
where $\widetilde{\bm{y}}_k^{-}$ is the innovation process \cref{eq:innovationProcess}, whose linear form and its covariance, denoted by $P_{\widetilde{y}_k^-} $, are given as \cite{Oguri2024}
\begin{align}\begin{aligned}
\widetilde{\bm{y}}_k^{-} = C_k\widetilde{\bm{x}}_k + D_k\bm{w}_{\mathrm{obs} ,k},\quad
P_{\widetilde{y}_k^-} \triangleq \cov{\widetilde{\bm{y}}_k^{-}} = C_k \widetilde{P}_k^{-} C_k^\top + D_k D_k^\top.
\label{eq:innovationProcessLinear}
\end{aligned}\end{align}

\subsubsection{Linear Flight-path Control Model}

We consider a linear FPC policy to model \cref{eq:SOCpolicy} for convex formulation.
In particular, we use the following form of the control policy:
\begin{align}\begin{aligned}
\bm{u}_k &= \lbar{\bm{u}}_k + K_{k}\bm{z}_k,
\label{eq:SOClinPolicy}
\end{aligned}\end{align}
where $\lbar{\bm{u}}_k  $ is a nominal control input, $K_k$ is a feedback gain matrix, and $\bm{z}_k$ is a stochastic process given by
\begin{align}\begin{aligned}
\bm{z}_{k+1} = A_k\bm{z}_k + L_{k+1}\widetilde{\bm{y}}_{k+1}^-,
\quad
\bm{z}_0 = \hat{\bm{x}}_0 - \lbar{\bm{x}}_0.
\label{eq:z-process}
\end{aligned}\end{align}

\subsubsection{State and Control Statistics}
We are now ready to analytically express the statistics of the state and control---in particular, their first and second statistical moments, i.e., mean and covariance.
As derived in Ref.~\cite{Oguri2024}, it is convenient to express \cref{eq:linEstProcess} in a block-matrix form as:
\begin{align}\begin{aligned}
&
\begin{bmatrix}
\hat{\bm{x}}_0 \\
\hat{\bm{x}}_1 \\
\hat{\bm{x}}_2 \\
\vdots
\end{bmatrix}
=
\begin{bmatrix}
I_{n_x} \\
A_0 \\
A_1 A_0 \\
\vdots
\end{bmatrix}
\hat{\bm{x}}_0^{-}
+
\begin{bmatrix}
0   & 0 & \\
B_0                 & 0 & \\
A_1 B_0             & B_1               & \\
					&                   & {\ddots}
\end{bmatrix}
\begin{bmatrix}
\bm{u}_0 \\
\bm{u}_1 \\
\bm{u}_2 \\
\vdots
\end{bmatrix}
+
\begin{bmatrix}
0   & 0 & \\
I_{n_x} 		 	& 0 & \\
A_1                 & I_{n_x} 			& \\
					&                   & {\ddots}
\end{bmatrix}
\begin{bmatrix}
\bm{c}_0 \\
\bm{c}_1 \\
\bm{c}_2 \\
\vdots
\end{bmatrix}
+
\begin{bmatrix}
L_0  				& 0 & \\
A_0L_0  			& L_1 				& \\
A_1A_0L_0 			& A_1L_1			& \\
					&                   & {\ddots}
\end{bmatrix}
\begin{bmatrix}
\widetilde{\bm{y}}_0^{-} \\
\widetilde{\bm{y}}_1^{-} \\
\widetilde{\bm{y}}_2^{-} \\
\vdots
\end{bmatrix}
\label{eq:blockMatrixWrittenOut}
\end{aligned}\end{align}
which can be expressed in a compact form as:
\begin{align}\begin{aligned}
\hat{\bm{X}} = 
\mathbf{A}\hat{\bm{x}}_0^{-} + 
\mathbf{B}\bm{U} + 
\mathbf{C} + 
\mathbf{L}{\bm{Y}},
\label{eq:blockMatrixEq}
\end{aligned}\end{align}
where 
$\hat{\bm{X}} = [\hat{\bm{x}}_0^\top, \hat{\bm{x}}_1^\top, ..., \hat{\bm{x}}_N^\top]^\top $,
$\bm{U} = [\bm{u}_0^\top, \bm{u}_1^\top, ..., \bm{u}_{N-1}^\top]^\top $,
$\bm{Y} = [\widetilde{\bm{y}}_0^{-1^\top}, \widetilde{\bm{y}}_1^{-1^\top}, ..., \widetilde{\bm{y}}_{N}^{-1^\top}]^\top $,
and $\mathbf{A},\mathbf{B},\mathbf{C}, \mathbf{L}$ are defined accordingly.
Similarly, define $\mathbf{K}$ to contain $K_k$ in a block matrix form as in Ref.~\cite{Oguri2024}.
Let us also define matrices $E_{x_k}$ and $E_{u_k} $ to extract $\bm{x}_k$ and $\bm{u}_k$ from $\bm{X} $ and $\bm{U} $, respectively, as $\bm{x}_k = E_{x_k} \bm{X}$ and $\bm{u}_k = E_{u_k} \bm{U}$.
Then, the mean and covariance of the state are derived in Ref.~\cite{Oguri2024}, summarized as follows:
\begin{align}\begin{aligned}
\lbar{\bm{x}}_k &= E_{x_k} (\mathbf{A}\lbar{\bm{x}}_0 + \mathbf{B}\lbar{\bm{U}} + \mathbf{C}),
\quad && (\text{mean}) 
\\
P_k &= \hat{P}_k + \widetilde{P}_k, 
\quad && (\text{total covariance})
\\
\hat{P}_k &= E_{x_k} (I + \mathbf{BK}) \mathbf{S}(I + \mathbf{BK})^\top E_{x_k}^\top,
\quad && (\text{state dispersion covariance})
\\
{P}_{u_k} &= E_{u_k} \mathbf{K} \mathbf{S}  \mathbf{K}^\top E_{u_k}^\top
\quad && (\text{control covariance})
\label{eq:statistics}
\end{aligned}\end{align}
where recall that $\widetilde{P}_k $ is the OD covariance and recursively calculated according to the Kalman filter as in \cref{eq:estProcessLin}.
$\mathbf{S}$ is a term that is independent from the control variables, given by
\begin{align}
\mathbf{S} = 
\mathbf{A}\hat{P}_{0^-}\mathbf{A}^\top + \mathbf{L}\mathbf{P}_Y \mathbf{L}^\top.
\label{eq:def-s}
\end{align}

In our formulation, a Cholesky square-root form of covariance matrices is convenient due to its convex property.
Such square-root forms of covariance matrices are given as follows:
\begin{align}\begin{aligned}
P_k^{1/2} =
\begin{bmatrix}
\hat{P}_k^{1/2} & \widetilde{P}_k^{1/2}
\end{bmatrix}
,\quad
\hat{P}_k^{1/2} &= E_{x_k} (I + \mathbf{BK}) \mathbf{S}^{1/2},
\quad
{P}_{u_k}^{1/2} = E_{u_k} \mathbf{K} \mathbf{S}^{1/2}
,\quad
\mathbf{S}^{1/2} = 
\begin{bmatrix}
\mathbf{A}\hat{P}_{0^-}^{1/2} & \mathbf{L}\mathbf{P}_Y^{1/2}
\label{eq:sqrtForm}
\end{bmatrix}
\end{aligned}\end{align}

\begin{remark}
\label{remark:affine}
As clear from \cref{eq:statistics,eq:sqrtForm}, $\lbar{\bm{x}}_k, P_k^{1/2}, \hat{P}_k^{1/2}, $ and $P_{u_k}^{1/2}$ are affine in the FPC variables $\lbar{\bm{u}}_k $ and $K_k$ (see Proposition 1 of Ref.~\cite{Oguri2024}), which plays a key role in our convex formulation.
\end{remark}

%!TEX root = main.tex

\subsection{Cost Function}
Under ZoH control input $\bm{u}(t) = \bm{u}_k,\ \forall t\in[t_k, t_{k+1}) $, \cref{eq:SOCobjective99fuel} is equivalent to 
\begin{align}
J = 
\sum_{k=0}^{N-1}Q_{X\sim\norm{\bm{u}_k}_2}(p) \Delta t_k
,\quad 
\Delta t_k = t_{k+1} - t_k
\label{eq:SOCobjective99fuelDiscrete}
\end{align}
Applying Lemma 4 of Ref.~\cite{Oguri2024}, we can calculate the upper bound of \cref{eq:SOCobjective99fuelDiscrete} as follows:
\begin{align}\begin{aligned}
J \leq J_{\mathrm{ub}} = \sum_{k} \left[ \norm{\lbar{\bm{u}}_k}_2 + m_{\chi^2}\left(1-p,n_u\right) \norm{P_{u_k}^{1/2} }_2 \right]\Delta t_k 
\label{eq:convexSOCcost}
\end{aligned}\end{align}
where $m_{\chi^2}(\varepsilon,n_u)$ is the square root of the quantile function of the chi-squared distribution with $n_u$ degrees of freedom, evaluated at probability $(1 - \varepsilon)$, and is mathematically expressed as:
\begin{align}\begin{aligned}
m_{\chi^2}(\varepsilon,n_u) = \sqrt{Q_{X\sim \chi^2(n_u)}(1 - \varepsilon)}
\label{eq:chi2inv}
\end{aligned}\end{align}
Importantly, \cref{eq:convexSOCcost} is convex in $\lbar{\bm{U}}$ and $\mathbf{K}$ due to \cref{remark:affine}.
We minimize the upper bound $J_{\mathrm{ub}}$ instead of $J$ directly.
$m_{\chi^2}(\varepsilon,n_u)$ is straightforward to calculate in modern programming languages.
We can calculate $m_{\chi^2}(\varepsilon,n_u)$ via $\sqrt{\mathtt{chi2inv}(1 - \varepsilon, n_u)} $ in Matlab and $\sqrt{\mathtt{chi2.cdf}(1 - \varepsilon, n_u)}$ in \texttt{scipy}.

\subsection{Constraint Function}
\label{sec:constraints-convex}

Tractable formulation of \cref{eq:SOCconstraints} depends on the specific form of constraints of interest.
In this article, we consider three important constraints that are common in interplanetary low-thrust trajectory design:
\textit{thrust magnitude constraint}, \textit{terminal constraints}, and \textit{planetary gravity assist constraint}.

\subsubsection{Thrust Magnitude Chance Constraints:}
With a risk tolerance $\varepsilon_u$ and maximum thrust acceleration $u_{\max}>0$,  the thrust magnitude constraint is expressed as a chance constraint as:
\begin{align} \begin{aligned}
\P{\norm{\bm{u}_k}_2 \leq u_{\max} } \geq 1 - \varepsilon_u,\ k=0,1,...,N-1
\label{eq:stoControlConst}
\end{aligned} \end{align}
Since \cref{eq:stoControlConst} is a special form of Eq.~53 in Ref.~\cite{Oguri2024}, we can directly apply Lemma 3 of Ref.~\cite{Oguri2024} to \cref{eq:stoControlConst}, yielding
\begin{align}
\norm{\lbar{\bm{u}}_k}_2 + m_{\chi^2}\left(\varepsilon_u,n_u\right)\norm{P_{u_k}^{1/2} }_2
\leq u_{\max},\ k=0,1,...,N-1
\label{eq:stoControlConstDetFormConvex}
\end{align}
which is again convex in $\lbar{\bm{U}}$ and $\mathbf{K}$ due to \cref{remark:affine}.
Note that convex formulation of thrust magnitude chance constraints with mass uncertainty requires more analytical effort; see Ref.~\cite{Oguri2022f} for instance.

\begin{remark}
Control chance constraints given by \cref{eq:stoControlConstDetFormConvex} provide a better approximation than the one proposed in Ridderhof et.al.\cite{Ridderhof2020c} (and used in Benedikter et.al.\cite{Benedikter2023}).
The key difference lies in the definition of $m_{\chi^2}\left(\varepsilon_u,n_u\right) $ that appear in \cref{eq:stoControlConstDetFormConvex,eq:convexSOCcost}, which is defined as in \cref{eq:chi2inv} in this study.
On the other hand, the equivalent of $m_{\chi^2}$ proposed in previous studies \cite{Ridderhof2020c,Benedikter2023} takes the following form (denoting it as $m_{\chi^2\mathrm{(old)}}$):
\begin{align}\begin{aligned}
m_{\chi^2\mathrm{(old)}} = 
\begin{cases}
\sqrt{2\ln{\frac{1}{\varepsilon}}} + \sqrt{n_u} 	&\quad (n_u > 2), \\
\sqrt{2\ln{\frac{1}{\varepsilon}}} 				&\quad (n_u = 1, 2)
\end{cases}
\label{eq:sgm_Jack}
\end{aligned}\end{align}
The values of $m_{\chi^2}$ given by \cref{eq:chi2inv} (proposed one) and by \cref{eq:sgm_Jack} (previous one \cite{Ridderhof2020c,Benedikter2023}) coincide when $n_u=2$, whereas those by \cref{eq:chi2inv} are smaller for $n_u>2$, leading to tighter (hence better) constraints.
This tighter bound is particularly beneficial in space applications since we typically deal with three-dimensional control input (e.g., thrust, acceleration, etc).
\cref{t:comparison} numerically compares the two approaches, confirming the tightness of our approximation.
The benefit of this tighter bound is not limited to control magnitude constraints but broadly applicable to any magnitude-related chance constraints, including state deviation constraints (see Ref.~\cite{Oguri2024}).
\end{remark}

\begin{table}[tb]
\centering
\caption{Magnitude constraint tightness comparison}
\small
\label{t:comparison}
\begin{tabular}{lrccccc}
\hline\hline
Approach & $\{\varepsilon,n_u\}=$				& $\{10^{-2},3\} $ & $\{10^{-3},3\} $ & $\{10^{-2},4\} $ & $\{10^{-3},4\} $ \\ \hline
Proposed [\cref{eq:chi2inv}] & $m_{\chi^2}=$  										& $3.3682  $ 			& $4.0331  $ & $3.6437 $ & $4.2973 $  \\
Previous \cite{Ridderhof2020c} [\cref{eq:sgm_Jack}] & $m_{\chi^2\mathrm{(old)}}=$  			& $4.7669  $ 			& $5.4490  $ & $5.0349 $ & $5.7169 $  \\
\hline\hline
\end{tabular}
\end{table}

\subsubsection{Terminal Constraints}

A convex formulation of \cref{eq:SOCterminal} is readily available in existing studies (e.g., Refs.~\cite{Ridderhof2020,Oguri2024}), yielding:
\begin{subequations}
\begin{align}
&E_{x_N}(\mathbf{A}\lbar{\bm{x}}_0 + \mathbf{B}\lbar{\bm{U}} + \bm{C}) - \bm{x}_f = 0
\label{eq:convexTerminalMean}
,\\
&\norm{(P_f - \widetilde{P}_N)^{-1/2}E_{x_N} (I + \mathbf{BK})\mathbf{S}^{1/2} }_2 - 1 \leq 0.
\label{eq:convexTerminalCovariance}
\end{align}
\label{eq:convexTerminal}
\end{subequations}
which are convex in $\lbar{\bm{U}}$ and $\mathbf{K}$.

\subsection{Incorporating Planetary Gravity Assists}
\label{sec:GA}

As commonly done in interplanetary mission design, we use the patched-conic model \cite{Sims1996} to model the effect of GA on the state trajectory under uncertainty and formulate (chance) constraints accordingly.
In this model, a GA rotates the spacecraft velocity relative to the planet while the periapsis radius about the planet, $r_{\mathrm{peri}} $ is constrained to be above a planetary impact radius, $r_{p\min}$.
Note that a higher-fidelity GA modeling is also possible to incorporate in the proposed framework.
In that case, extending the formulation proposed in Ref.~\cite{Ellison2019} to stochastic setting would be a straightforward approach, where the state statistics would be propagated in the planet-centric frame while imposing the impact constraint.

Let a GA occur at $t_k$, and assign $k$ and $k+1$ to the nodes that are immediately before and after the GA, respectively.
Denoting the spacecraft velocity relative to the planet by $\bm{v}^{\infty}_{k}= \bm{v}_k - \bm{v}_p(t_k) $, where $\bm{v}_p(t_k) $ is the planet's velocity at $t_k$, under the patched-conic model, the GA effect is instantaneous and rotates $\bm{v}^{\infty}_{k}$ by a turn angle $\theta$ \cite{Sims1996}, i.e.,
\begin{align}\begin{aligned}
t_k = t_{k+1},
\quad
\bm{r}_k = \bm{r}_{k+1},
\quad
\norm{\bm{v}^{\infty}_{k+1}}_2 = \norm{\bm{v}^{\infty}_k}_2,
% \norm{\bm{v}^{\infty}_{k+1}}_2 = \norm{\bm{v}^{\infty}_k}_2 = v_{\infty},
\quad
\bm{v}^{\infty}_{k+1} \cdot \bm{v}^{\infty}_k = \norm{\bm{v}^{\infty}_k}_2^2 \cos{\theta}
\label{eq:GAcondition}
\end{aligned}\end{align}
The periapsis radius with respect to the planet during the flyby, $r_{\mathrm{peri}}(\cdot)$, is a function of $\theta$ and $\bm{v}^{\infty}_{k}$ (hence of $\bm{v}_k$ and $t_k$), and must be greater than a minimum altitude, $r_{p\min}$, as follows \cite{Sims1996}:
\begin{align}\begin{aligned}
r_{\mathrm{peri}}(\bm{v}_k, \theta, t_k) \geq r_{p\min},\quad
r_{\mathrm{peri}}(\bm{v}_k, \theta, t_k) = \frac{\mu_p}{\norm{\bm{v}^{\infty}_k}_2^2}\left(\frac{1}{\sin{\frac{\theta}{2}}} - 1\right),
\label{eq:impactConstraint0}
\end{aligned}\end{align}
where $\mu_p$ is the planet's gravity constant.
To ease the formulation in \cref{sec:impact-CC}, we equivalently express \cref{eq:impactConstraint0} as:
\begin{align}\begin{aligned}
\norm{\bm{v}^{\infty}_k}_2 - \sqrt{\frac{\mu_p}{r_{p\min}}\left(\frac{1}{\sin{\frac{\theta}{2}}} - 1 \right) } \leq 0
\label{eq:impactConstraintDet}
\end{aligned}\end{align}

\subsubsection{State Statistics Mapping across a Gravity Assist}
While \cref{eq:GAcondition} is simple to apply for deterministic problems, it is not suited for stochastic problems because it does not clarify how the state statistics should be mapped across a GA.
Thus, this study derives another form of \cref{eq:GAcondition} that furnishes an explicit formula that maps $P_k $ to $P_{k+1}$ across a GA event.
To this end, we express the rotation of $\bm{v}_{\infty}$ via an orthogonal matrix $R_{\mathrm{GA}}\in\R^{3\times3}$ as $\bm{v}^{\infty}_{k+1} = R_{\mathrm{GA}} \bm{v}^{\infty}_k $, where $R_{\mathrm{GA}}$ is assumed to be deterministic (hence $\theta$ too) and may be considered as an optimization variable.
However, using $R_{\mathrm{GA}}$ as an optimization variable introduces another issue;
imposing the orthogonality condition is non-convex.
To address this issue, this study utilizes \textit{Cayley transform}\footnote{For any skew-symmetric matrix $A$, a matrix $Q$ that satisfies $Q=(I +A)^{-1} (I - A) $ is an orthogonal matrix with its determinant $+1$, i.e., $Q Q^\top = I$ and $|Q| = 1 $.} to express the orthogonal matrix $R_{\mathrm{GA}}$ in terms of a skew-symmetric matrix $V\in\R^{3\times3}$ (hence $V^\top = - V$) as $R_{\mathrm{GA}} = (I_3 + V)^{-1} (I_3 - V)$, yielding
\begin{align}\begin{aligned}
\bm{v}^{\infty}_{k+1} = (I_3 + V)^{-1} (I_3 - V) \bm{v}^{\infty}_k
\quad \Leftrightarrow \quad
\bm{v}_{k+1} - \bm{v}_p(t_k) = (I_3 + V)^{-1} (I_3 - V) [\bm{v}_{k} - \bm{v}_p(t_k)]
\label{eq:CayleyVinf}
\end{aligned}\end{align}
where $I_n\in\R^{n\times n} $ is an identity matrix.
It is clear that, because $R_{\mathrm{GA}} = (I_3 + V)^{-1} (I_3 - V)$ is a rotation matrix, \cref{eq:CayleyVinf} automatically satisfies the third constraint of \cref{eq:GAcondition}.

We can parameterize the skew-symmetric matrix $V$ by only three scalars and express it as $V = [\bm{u}_k]_{\times}$ with $\bm{u}_k\in\R^3$, where $[\cdot]_{\times}$ denotes the matrix cross product operator.
Here, we can see $\bm{u}_k$ as a control input at $t_k$ that characterizes the effect of the GA at this epoch.
This viewpoint, combined with the Cayley transformation, allows us to describe the mapping of state under the effect of a GA as a nonlinear function, which then provides a linear mapping equation.
This idea is formally stated in \cref{lemma:GAequation}.

\begin{lemma}
\label{lemma:GAequation}
Using Cayley transformation,
the mapping of orbital states $\bm{x}_{k}$ to $\bm{x}_{k+1}$ due to a gravity assist is expressed as a nonlinear function of $\bm{x}_k, \bm{u}_k, t_k$, as follows:
\begin{align}
\bm{x}_{k+1} = \bm{f}_{\mathrm{GA}}(\bm{x}_k, \bm{u}_k, t_k)
=
\begin{bmatrix}
\bm{r}_k \\
\bm{v}_p(t_k) + R_{\mathrm{GA}}(\bm{u}_k) \cdot [\bm{v}_k - \bm{v}_p(t_k)]
\end{bmatrix}
,
\quad
R_{\mathrm{GA}}(\bm{u}_k) = (I_3 + [\bm{u}_k]_{\times})^{-1} (I_3 - [\bm{u}_k]_{\times})
\label{eq:NLeqGA}
\end{align}
where $[\cdot]_{\times}$ denotes the matrix cross product operator, and $(I_3 + [\bm{u}_k]_{\times})^{-1}$ always exists.
Also, its linearized equation can be expressed as $\bm{x}_{k+1} = A_k\bm{x}_k + B_k \bm{u}_k + \bm{c}_k $, where the linear system matrices $A_k, B_k, \bm{c}_k$ are given by
\begin{align}
\begin{aligned}
A_k &= 
\left.\frac{\partial \bm{f}_{\mathrm{GA}}}{\partial \bm{x}}\right|^{\mathrm{ref}}
= 
\begin{bmatrix}
I_3 & 0_{3\times3} \\
0_{3\times3} & R_{\mathrm{GA}}(\bm{u}_k^{\mathrm{ref}})
\end{bmatrix},
\quad
B_k = 
\left.\frac{\partial \bm{f}_{\mathrm{GA}}}{\partial \bm{u}}\right|^{\mathrm{ref}}
= 
\begin{bmatrix}
0_{3\times3} \\
(I_3 + [\bm{u}_k^{\mathrm{ref}}]_{\times})^{-1} ([{\bm{v}}_{k+1}^{\mathrm{ref}}]_{\times} + [{\bm{v}}_{k}^{\mathrm{ref}}]_{\times} - 2[{\bm{v}}_p]_{\times})
\end{bmatrix}
,\\
\bm{c}_k &= 
\bm{f}_{\mathrm{GA}}(\bm{x}_k^{\mathrm{ref}}, \bm{u}_k^{\mathrm{ref}}) - A_k \bm{x}_k^{\mathrm{ref}} - B_k \bm{u}_k^{\mathrm{ref}}
\end{aligned}
\label{eq:LinEqGA}
\end{align}
\end{lemma}

\begin{proof}
See \cref{proof:GAequation}.
\end{proof}

Using \cref{lemma:GAequation}, we can integrate the effect of GAs into the form of \cref{eq:linEstProcess} with $L_{k+1} = 0$, assuming that we do not have the opportunity to perform measurement updates between $t_k$ and $t_{k+1}$ since $t_k = t_{k+1} $ at a GA.
This way, we can capture the statistical effect of GAs in a unified form, and hence the state statistics mapping due to GA are characterized by \cref{eq:statistics,eq:sqrtForm} with $A_k, B_k, \bm{c}_k$ given as in \cref{eq:LinEqGA}.
Note here that, since the turn angle $\theta$ (and hence $V$) is assumed deterministic, $\bm{u}_k$ at the GA epoch is also deterministic, and therefore $\lbar{\bm{u}}_k = \bm{u}_k$ and $K_k = 0$.

The development above takes care of the first three equations of \cref{eq:GAcondition} in the stochastic setting.
\cref{sec:impact-CC,sec:turn-const} discuss the reformulation of the last equation of \cref{eq:GAcondition} and \cref{eq:impactConstraintDet}.

\subsubsection{Turn angle constraint}
\label{sec:turn-const}

The turn angle constraint given by the last equation of \cref{eq:GAcondition} is a deterministic constraint in the current formulation due to the assumption that the turn angle $\theta$ is deterministic.
Hence, the turn angle constraint needs to be satisfied at the mean values of $\bm{v}^{\infty}_{k+1}$ and $\bm{v}^{\infty}_k$, i.e., 
\begin{align}\begin{aligned}
g_{\mathrm{GA-turn}}(\lbar{\bm{x}}_{k+1}, \lbar{\bm{x}}_k, \theta) = 
\norm{\lbar{\bm{v}}^{\infty}_k}_2^2 \cos{\theta} - \lbar{\bm{v}}^{\infty}_{k+1} \cdot \lbar{\bm{v}}^{\infty}_k = 0,
\quad
\lbar{\bm{v}}^{\infty}_k = E_v\bm{x}_k - \bm{v}_p(t_k)
\label{eq:turn-angle}
\end{aligned}\end{align}
where $E_v = [0_{3\times 3}\ \ I_3]$ extracts the velocity from $\bm{x}_k$.
\cref{eq:turn-angle} is clearly non-convex in $\lbar{\bm{x}}_{k+1}, \lbar{\bm{x}}_k, \theta$ and thus linearized about $\lbar{\bm{x}}_k^{\mathrm{ref}}, \theta^{\mathrm{ref}}$ to form convex subproblem at each iteration, where $(\cdot)^{\mathrm{ref}}$ indicates evaluation at the solution of the previously accepted iteration.
Indicating the linearized equation with $\widetilde{(\cdot)}$, linearizing \cref{eq:turn-angle} about $\lbar{\bm{x}}_{k+1}^{\mathrm{ref}}, \lbar{\bm{x}}_k^{\mathrm{ref}}, \theta^{\mathrm{ref}}$ leads to
{\footnotesize
\begin{align}\begin{aligned}
\widetilde{g}_{\mathrm{GA-turn}}(\lbar{\bm{x}}_{k+1}, \lbar{\bm{x}}_k, \theta) &= 
g_{\mathrm{GA-turn}}^{\mathrm{ref}} 
+ \left.\frac{\partial g_{\mathrm{GA-turn}}}{\partial \lbar{\bm{x}}_{k+1}}\right|^{\mathrm{ref}}(\lbar{\bm{x}}_{k+1} - \lbar{\bm{x}}_{k+1}^{\mathrm{ref}})
+ \left.\frac{\partial g_{\mathrm{GA-turn}}}{\partial \lbar{\bm{x}}_k}\right|^{\mathrm{ref}}(\lbar{\bm{x}}_k - \lbar{\bm{x}}_k^{\mathrm{ref}})
+ \left.\frac{\partial g_{\mathrm{GA-turn}}}{\partial \theta}\right|^{\mathrm{ref}}(\theta - \theta^{\mathrm{ref}})
,
\\
\frac{\partial g_{\mathrm{GA-turn}}}{\partial \lbar{\bm{x}}_k} &=
2 \lbar{\bm{v}}^{\infty^\top}_k E_v \cos{\theta} - \lbar{\bm{v}}^{\infty^\top}_{k+1} E_v
,\quad
\frac{\partial g_{\mathrm{GA-turn}}}{\partial \lbar{\bm{x}}_{k+1}} =
- \lbar{\bm{v}}^{\infty^\top}_{k} E_v
,\quad
\frac{\partial g_{\mathrm{GA-turn}}}{\partial \theta} =
- \norm{\lbar{\bm{v}}^{\infty}_k}_2^2 \sin{\theta}
\label{eq:turnAngleConst-linear}
\end{aligned}\end{align}}

\subsubsection{Impact chance constraint}
\label{sec:impact-CC}

Due to the stochasticity in $\bm{v}_k$ and hence in $\bm{v}^{\infty}_k$, the impact constraint, \cref{eq:impactConstraintDet}, is ill-defined in its current form, and needs to be expressed as a chance constraint with a risk level $\varepsilon_{\mathrm{GA}}$:
\begin{align}\begin{aligned}
\P{\norm{\bm{v}^{\infty}_k}_2 - \sqrt{\frac{\mu_p}{r_{p\min}}\left(\frac{1}{\sin{\frac{\theta}{2}}} - 1 \right) } \leq 0} \geq 1 - \varepsilon_{\mathrm{GA}}
\label{eq:stoCCperiapsis}
\end{aligned}\end{align}
Noting that $\bm{v}^{\infty}_k = \bm{v}_k - \bm{v}_p(t_k) = E_v\bm{x}_k - \bm{v}_p(t_k) \sim \mathcal{N}(E_v\lbar{\bm{x}}_k - \bm{v}_p,\ E_v P_k E_v^\top) $, we can take the same approach as for \cref{eq:stoControlConst} and directly apply Lemma 3 of Ref.~\cite{Oguri2024} to \cref{eq:stoCCperiapsis}, leading to
\begin{align}
g_{\mathrm{GA-impact}}(\lbar{\bm{x}}_k, P_k^{1/2}, \theta) = \norm{E_v\lbar{\bm{x}}_k - \bm{v}_p}_2 + m_{\chi^2}\left(\varepsilon_{\mathrm{GA}},3\right)\norm{E_v P_k^{1/2}}_2 - \sqrt{\frac{\mu_p}{r_{p\min}}\left(\frac{1}{\sin{\frac{\theta}{2}}} - 1 \right) } \leq 0
\label{eq:detCCperiapsis}
\end{align}
which is convex in $\lbar{\bm{u}}_k, K_k, \forall k $ due to \cref{remark:affine}, but not convex in $\theta$.
Thus, $g_{\mathrm{GA-impact}}(\cdot)$ is linearized about $\theta^{\mathrm{ref}}$, yielding
\begin{align}
\widetilde{g}_{\mathrm{GA-impact}}(\lbar{\bm{x}}_k, P_k^{1/2}, \theta) =
g_{\mathrm{GA-impact}}(\lbar{\bm{x}}_k, P_k^{1/2}, \theta^{\mathrm{ref}}) - \sqrt{\frac{\mu_p}{r_{p\min}}}\ \frac{\partial }{\partial \theta}\left(\sqrt{\frac{1}{\sin{\frac{\theta}{2}}} - 1} \right)_{\theta=\theta^{\mathrm{ref}}} (\theta - \theta^{\mathrm{ref}})
\label{eq:detCCperiapsis-linear}
\end{align}

%!TEX root = main.tex
\section{Solution Method via Sequential Convex Programming}
\label{sec:solutionMethod}

\subsection{Penalized Convex Subproblem with Trust Region}
\label{sec:penalty-formulation}

Like typical SCP approaches for nonlinear trajectory optimization, we introduce \textit{virtual buffers} and \textit{trust region} and penalize the constraint buffers to facilitate the constraint satisfaction;
see \cite{Mao2016,Oguri2023b,Oguri2024c} for more detail.

Virtual buffers are typically introduced to relax constraints that are originally non-convex and hence need to be approximated for forming a convex subproblem.
On the other hand, relaxing all the approximated constraints can lead to too loose problem formulations.
In general, the choice of which constraints to relax highly depends on the problem and reflects the mission designers' experience and insights.
Denote the equality and inequality constraints to relax as $\bm{g}^{\mathrm{eq, relax}} (=0)$ and $\bm{g}^{\mathrm{ineq, relax}} (\leq 0)$, respectively, and their approximated convex forms as $\widetilde{\bm{g}}^{\mathrm{eq, relax}}$ and $\widetilde{\bm{g}}^{\mathrm{ineq, relax}}$.
Then, the approximated constraints with relaxation via virtual buffers are expressed as:
\begin{align}
\widetilde{\bm{g}}^{\mathrm{eq, relax}}(\lbar{\bm{X}},\bar{\bm{U}},\mathbf{K},\bm{\Theta}) = \bm{\xi}, \quad
\widetilde{\bm{g}}^{\mathrm{ineq, relax}}(\lbar{\bm{X}},\bar{\bm{U}},\mathbf{K},\bm{\Theta}) \leq \bm{\zeta}, 
\label{eq:relaxedConstraints}
\end{align}
where $\bm{\xi} $ and $\bm{\zeta} \geq 0 $ are virtual buffers for equality and inequality constraints, respectively.
Recall that $\bm{\Theta}$ denotes a vector of relevant parameters to optimize (e.g., turn angle $\theta$ at each GA).

Specifically, in this work, we relax the final mean state constraint (equality) and the gravity-assist impact chance constraint (inequality), which are both non-convex in variables in the original form.
Hence, $\widetilde{\bm{g}}^{\mathrm{eq, relax}}$ consists of \cref{eq:convexTerminalMean} and $\widetilde{\bm{g}}^{\mathrm{ineq, relax}}$ of \cref{eq:detCCperiapsis-linear}.

We formulate the penalty function based on an SCP algorithm called \texttt{SCvx*} (\textit{SCvx-star}) \cite{Oguri2023b}, which leverages the augmented Lagrangian framework to extend its predecessor \texttt{SCvx} \cite{Mao2016,Mao2018} to provide a theoretical convergence guarantee to a feasible optimal solution.
\texttt{SCvx*} introduces Lagrange multipliers $\bm{\lambda} $ and $\bm{\mu} \geq 0$, which have the same dimensions as the virtual buffers $\bm{\xi}$ and $\bm{\zeta}$, respectively, and defines the penalty function as:
\begin{align}
P (\bm{\xi}, \bm{\zeta}) = \sum_{i}
\lambda_i \xi_i + \frac{1}{w}\phi(w\xi_i)
+
\mu_i \zeta_i + \frac{1}{w}\phi(w[\zeta_i]_+)
\label{eq:ALpenalty}
\end{align}
where $w>0$ is a penalty weight, $\phi(\cdot): \R\mapsto \R_+$ is strictly convex and smooth (once continuously differentiable), and $[\zeta_i]_+ = \max(\zeta_i, 0) $.
The multipliers $\lambda_i$ and $\mu_i$ are updated according to the augmented Lagrangian framework (in its generalized form; see Chapter 5, \cite{Bertsekas1982}):
\begin{align}
\lambda_i \gets \lambda_i + \nabla \phi(w g^{\mathrm{eq, relax}}_i),
\quad
\mu_i \gets \left[\mu_i + \nabla \phi(w g^{\mathrm{ineq, relax}}_i) \right]_+ ,
\label{eq:LagrangeUpdate}
\end{align}
where $g^{\mathrm{eq, relax}}_i $ and $g^{\mathrm{ineq, relax}}_i $ represent the $i$-th elements of the original functions for the relax equality and inequality constraints, respectively.
If we use $\phi(z) = z^2/2$, then $\nabla \phi = z$ (where $z\in\R$ is a dummy variable), which reduces \cref{eq:LagrangeUpdate} to the conventional multiplier updates: $\lambda_i \gets \lambda_i + w g^{\mathrm{eq, relax}}_i$ and $\mu_i \gets [\mu_i + wg^{\mathrm{ineq, relax}}_i ]_+$ (see, for instance, \cite{Oguri2023b,Lantoine2012}).

On the other hand, this study utilizes the following form of $\phi(\cdot)$ to define the penalty function: 
\begin{align}
\phi(z) = \frac{1}{\tau} \abs{z}^\tau + \frac{1}{2}z^2,
\quad
\nabla \phi(z) = \mathrm{sign}(z) \abs{z}^{\tau-1} + z,
\quad \tau \in(1,2)
\label{eq:penalty-l1approx}
\end{align}
which is convex in $z$ and closely approximates $l_1$-norm penalty ($\phi(z) = \abs{z}$) near the origin ($z\approx 0$) when $\tau$ is close to unity. 
\cref{eq:penalty-l1approx} is chosen to take advantage of the desirable property of $l_1$-norm penalty (known as a class of \textit{exact} penalty functions in constrained optimization literature \cite{Han1979,Mayne1987}) while retaining the smoothness;
the smoothness (continuous differentiability) is required in the augmented Lagrangian framework (Chapter 5, \cite{Bertsekas1982}), while the vanilla $l_1$-norm penalty $\phi(z) = \abs{z} $ is not smooth at the origin.

\cref{f:penalty-fcn} shows the values of the $l_1$-like penalty function and its gradient with different $\tau$ given by \cref{eq:penalty-l1approx} and compares them against the $l_1$ exact penalty $\phi(z) = \abs{z}$ and quadratic penalty $\phi(z) = z^2$.
This comparison visually verifies that \cref{eq:penalty-l1approx} approaches the $l_1$ penalty near the origin as $\tau \to 1$ while retaining the smoothness, i.e., continuous $\nabla \phi(z)$, for any $z$ including the origin, as clear from \cref{f:grad}.
Based on these observations, the numerical examples in this paper employ $\tau = 1.1$.

\begin{figure}[tb]
\centering \subfigure[\label{f:phi} Function values]
{\includegraphics[width=0.47\linewidth]{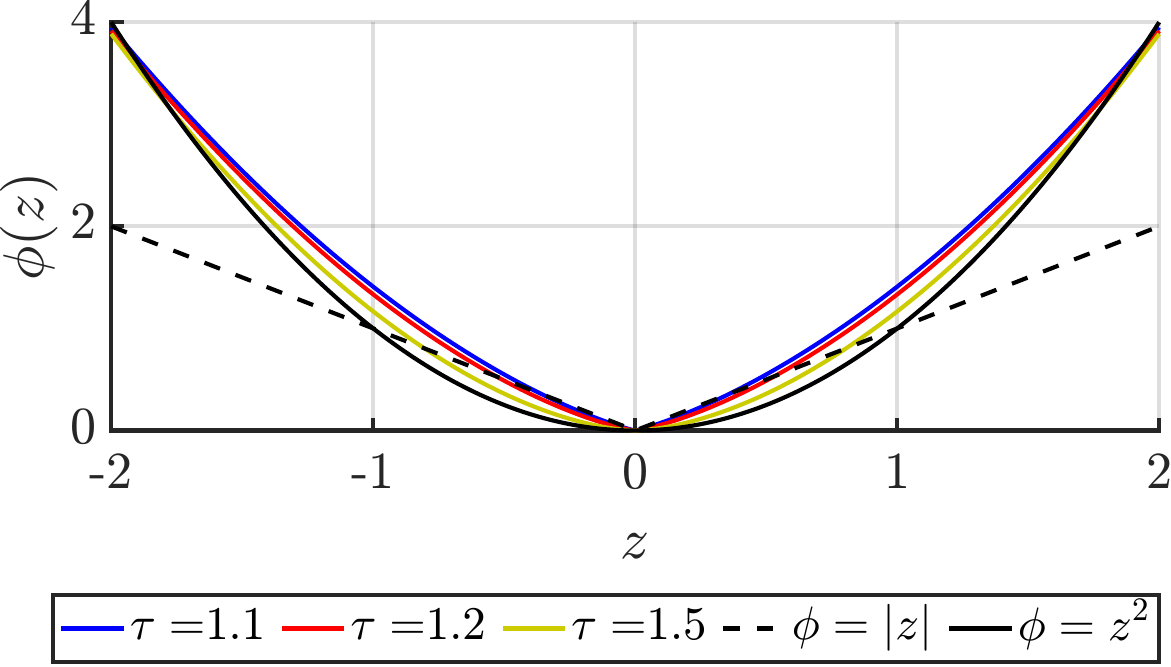}}
\centering \subfigure[\label{f:grad} Gradient values]
{\includegraphics[width=0.47\linewidth]{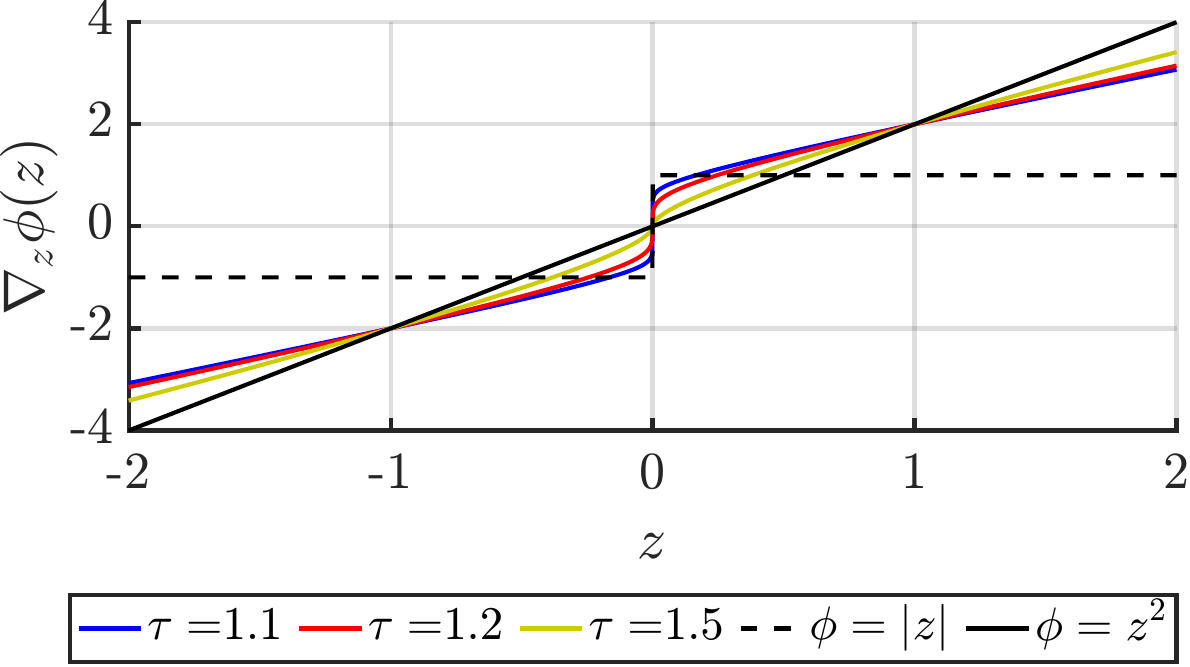}}
	\caption{\label{f:penalty-fcn} $l_1$-like smooth penalty function $\phi(z)$ given in \cref{eq:penalty-l1approx} and its gradient for different $\tau$ for $z\in[-2,2]$, compared with $l_1$ exact penalty function $\abs{z}$ and quadratic penalty $z^2$}
\end{figure}

We then argument the cost function ($\Delta$V99 upper bound) \cref{eq:convexSOCcost} by the augmented-Lagrangian penalty \cref{eq:ALpenalty} as:
\begin{align}
J_{\mathrm{aug}}(\lbar{\bm{X}},\lbar{\bm{U}}, \mathbf{K},\bm{\Theta}, \bm{\xi}, \bm{\zeta}) = J_{\mathrm{ub}}(\lbar{\bm{U}}, \mathbf{K}) + P (\bm{\xi}, \bm{\zeta})
\label{eq:augmentedCost}
\end{align}
We also impose trust region constraints on the decision variables to avoid artificial unboundedness:
\begin{align}
\norm{D_x(\lbar{\bm{X}} - \lbar{\bm{X}}^{\mathrm{ref}})}_{\infty} \leq \Delta_{\mathrm{TR}},\quad
\norm{D_u(\lbar{\bm{U}} - \lbar{\bm{U}}^{\mathrm{ref}})}_{\infty} \leq \Delta_{\mathrm{TR}},\quad
\norm{D_\theta(\bm{\Theta} - \bm{\Theta}^{\mathrm{ref}})}_{\infty} \leq \Delta_{\mathrm{TR}},
\label{eq:trustRegion}
\end{align}
where $\Delta_{\mathrm{TR}} > 0$ is the trust region radius, $D_x $, $D_u $, and $D_{\theta}$ are scaling (typically diagonal) matrices for the state, control, and parameter updates, respectively;
intuitively, greater $D_x, D_u$ leads to a stricter trust region for a given $\Delta_{\mathrm{TR}}$.
We do not impose  a trust region on $\mathbf{K} $, as $\mathbf{K} $ only linearly affects the trajectory of (Cholesky factor of) covariance matrices in our formulation.

\begin{remark}
\label{remark:convexity-cost-and-TR}
Both \cref{eq:augmentedCost,eq:trustRegion} are convex in variables.
$J_{\mathrm{aug}}$ in \cref{eq:augmentedCost} is convex in $\lbar{\bm{U}}, \mathbf{K}$ as clear from \cref{eq:convexSOCcost} and \cref{remark:affine}.
$P (\bm{\xi}, \bm{\zeta})$ in \cref{eq:augmentedCost}, defined in \cref{eq:ALpenalty} with \cref{eq:penalty-l1approx}, is convex in $\bm{\xi}$ and $\bm{\zeta}$ because \cref{eq:penalty-l1approx} is convex (which proves $\phi(w\xi_i)$ being convex in $\bm{\xi}$), and the fact that $[\zeta_i]_+$ is non-negative and $\tau > 1$ implies the convexity of $\phi(w[\zeta_i]_+)$ in $\bm{\zeta}$ (recall that, if $g$ is convex and nonnegative, then $g(x)^p$ is convex for $p\geq 0$ \cite{Boyd2004}).
\cref{eq:trustRegion} is clearly convex in $\lbar{\bm{X}}$.
\end{remark}

Combing the discussion thus far, the convex subproblem is formulated as in \cref{prob:convexSOC}.
\begin{problem}[Convex subproblem]
\label{prob:convexSOC}
Find $\lbar{\bm{X}}^* $, $\lbar{\bm{U}}^* $, $\mathbf{K}^* $, $\bm{\Theta}^*$, $\bm{\xi}^*$, and $\bm{\zeta}^*$ that
minimize the augmented cost \cref{eq:augmentedCost} while satisfying 
the filtered state dynamics constraints \cref{eq:sqrtForm} including the GA effect \cref{eq:LinEqGA}, 
control magnitude constraints \cref{eq:stoControlConstDetFormConvex}, 
GA turn-angle constraints \cref{eq:turnAngleConst-linear},
terminal covariance constraints \cref{eq:convexTerminalCovariance}, 
relaxed constraints \cref{eq:relaxedConstraints} with \cref{eq:convexTerminalMean,eq:detCCperiapsis-linear}, and
trust region constraints \cref{eq:trustRegion}.
\end{problem}

\begin{remark}
\cref{prob:convexSOC} is convex in variables due to \cref{remark:affine}, the discussion in \cref{sec:constraints-convex}, and \cref{remark:convexity-cost-and-TR}.
\end{remark}

\subsection{Sequential Convex Programming via \texttt{SCvx*} Algorithm}
\label{sec:SOCalgorithm}

With \cref{prob:originalSOC} reformulated in a deterministic, convex form as in \cref{prob:convexSOC}, we are ready to efficiently solve the problem via sequential convex programming (SCP).
In particular, this study utilizes a recent SCP algorithm called \texttt{SCvx*} \cite{Oguri2023b}.
The \texttt{SCvx*} algorithm extends a successful SCP algorithm known as \texttt{SCvx} \cite{Mao2016,Mao2018} to provide a theoretical guarantee for convergence to a feasible optimal solution by leveraging the augmented Lagrangian framework \cite{Bertsekas1982}.

\begin{algorithm}[tb]
\caption{Robust Trajectory Optimization under Uncertainty via \texttt{SCvx*}}
\label{alg:SCvx*}
\begin{algorithmic}[1]
\State 
Generate the initial reference trajectory $\{\lbar{\bm{X}}^{\mathrm{ref}}, \lbar{\bm{U}}^{\mathrm{ref}}, \mathbf{K}^{\mathrm{ref}}(=0), \bm{\Theta}^{\mathrm{ref}}\} $ 
% \Comment{by solving deterministic trajectory optimization}

\State 
Initialize \texttt{SCvx*} parameters
\Comment{Line 1 of Ref.~\cite{Oguri2023b}}

\While {the convergence criteria not met}
\Comment{Line 2 of Ref.~\cite{Oguri2023b}}

\State 
Obtain linearized block system matrices $\mathbf{A}, \mathbf{B}, \mathbf{C}, \mathbf{L}, \mathbf{Y} $
\Comment{based on the current reference $\lbar{\bm{X}}^{\mathrm{ref}}, \lbar{\bm{U}}^{\mathrm{ref}} $}

\State 
$\{\lbar{\bm{X}}^*, \lbar{\bm{U}}^*, \mathbf{K}^*, \bm{\Theta}^*, \bm{\xi}^*, \bm{\zeta}^* \} \gets $ solve \cref{prob:convexSOC}

\If{acceptance conditions met}
\Comment{\cref{eq:acceptance}}
\label{line:acceptance}

\State 
$\{\lbar{\bm{X}}^{\mathrm{ref}}, \lbar{\bm{U}}^{\mathrm{ref}}, \mathbf{K}^{\mathrm{ref}}, \bm{\Theta}^{\mathrm{ref}} \} \gets \{\lbar{\bm{X}}^*, \lbar{\bm{U}}^*, \mathbf{K}^*, \bm{\Theta}^* \}$
\Comment{solution update}

\State
Update Lagrange multipliers 
\Comment{Line 13-15, Algorithm 1 of Ref.~\cite{Oguri2023b}, with \cref{eq:LagrangeUpdate}}

\EndIf

\State
Update trust region
\Comment{\cref{eq:TRupdate}}

\EndWhile
\State \Return $\{\lbar{\bm{X}}^{\mathrm{ref}}, \lbar{\bm{U}}^{\mathrm{ref}}, \mathbf{K}^{\mathrm{ref}}, \bm{\Theta}^{\mathrm{ref}}  \}$
\end{algorithmic}
\end{algorithm}

\cref{alg:SCvx*} presents the solution algorithm based on \texttt{SCvx*}.
For the detail of the algorithm and its convergence property, see Ref.~\cite{Oguri2023b} (and Ref.~\cite{Oguri2023a}, which corrects a few minor typographical errors).
The returned quantities at the convergence $\{\lbar{\bm{X}}^*, \lbar{\bm{U}}^*, \mathbf{K}^*, \bm{\Theta}^*  \}$ approximately solves \cref{prob:originalSOC}, where the propagation of state uncertainty is approximated as the linear propagation of state mean and covariance.
Extending the solution method to incorporate non-Gaussian uncertainty quantification while retaining the computational tractability is an important direction of ongoing research.
There have been promising studies developed to address the problem of chance-constrained control under non-Gaussian uncertainty, such as those via Gaussian mixture model \cite{Kumagai2024,Boone2022d} and polynomial chaos expansion \cite{Nakka2019,Greco2022}, which have been successfully demonstrated in simpler systems, single-maneuver planning, or without FPC policy optimization (i.e., open-loop control).

The SCP algorithm developed in this study is nearly identical to the original \texttt{SCvx*}, although there are a few aspects that warrant some discussion, including some algorithmic modifications to extend the algorithm to be applicable to trajectory optimization under uncertainty.
These modifications have been also employed in our recent paper \cite{Kumagai2025}.
These aspects are highlighted in the following.

\subsubsection{Objective Improvement Measure}

A key in \cref{alg:SCvx*} is the \textit{solution acceptance} criterion in line \ref{line:acceptance}.
Similar to \cite{Oguri2023b,Mao2016}, the solution acceptance relies on the ratio of the nonlinear cost reduction $\Delta J$ to the approximated one $\Delta L$:
\begin{align}
\rho = \frac{\Delta J}{\Delta L} = 
\frac{J_{\mathrm{aug}}^{\mathrm{NL}}(\lbar{\bm{X}}^{\mathrm{ref}},\lbar{\bm{U}}^{\mathrm{ref}}, \mathbf{K}^{\mathrm{ref}}, \bm{\Theta}^{\mathrm{ref}}) - J_{\mathrm{aug}}^{\mathrm{NL}}(\lbar{\bm{X}}^*,\lbar{\bm{U}}^*, \mathbf{K}^*, \bm{\Theta}^*)}
{J_{\mathrm{aug}}^{\mathrm{NL}}(\lbar{\bm{X}}^{\mathrm{ref}},\lbar{\bm{U}}^{\mathrm{ref}}, \mathbf{K}^{\mathrm{ref}}, \bm{\Theta}^{\mathrm{ref}}) - J_{\mathrm{aug}}(\lbar{\bm{X}}^*,\lbar{\bm{U}}^*, \mathbf{K}^*, \bm{\Theta}^*, \bm{\xi}^*, \bm{\zeta}^*)}
\label{eq:rho}
\end{align}
where $J_{\mathrm{aug}}$ is given in \cref{eq:augmentedCost} while $J_{\mathrm{aug}}^{\mathrm{NL}}$ represents the cost of the original non-convex problem with augmented Lagrangian.
Note that one must use the same values of $\omega, \bm{\lambda}, \bm{\mu}$ when calculating all the $J_{\mathrm{aug}}^{\mathrm{NL}}$ and $J_{\mathrm{aug}}$ in \cref{eq:rho}.

Evaluation of $J_{\mathrm{aug}}^{\mathrm{NL}}$ requires approximation.
Exact evaluation of $J_{\mathrm{aug}}^{\mathrm{NL}}$ based on \cref{prob:originalSOC} is intractable since it necessitates evaluating the original $\Delta$V99 cost \cref{eq:SOCobjective99fuel} and penalty costs associated with chance constraints \cref{eq:SOCconstraints} and terminal constraint \cref{eq:SOCterminal} under the nonlinear filtering process \cref{eq:SOCfiltering} driven by the nonlinear SDE \cref{eq:SDE} and measurements \cref{eq:SOCobservations}.
To approximately compute $J_{\mathrm{aug}}^{\mathrm{NL}}$ in a tractable manner, this work takes full advantage of the linear covariance-based formulation developed in \cref{sec:tractableFormulation}.
Specifically, we define $J_{\mathrm{aug}}^{\mathrm{NL}}$ as:
\begin{align} \begin{aligned}
&J_{\mathrm{aug}}^{\mathrm{NL}}(\lbar{\bm{X}},\lbar{\bm{U}}, \mathbf{K}, \bm{\Theta}) = J_{\mathrm{aug}}[\lbar{\bm{X}},\lbar{\bm{U}}, \mathbf{K}, \bm{g}^{\mathrm{eq,relaxed}}(\lbar{\bm{X}},\lbar{\bm{U}}, \mathbf{K}, \bm{\Theta}), \bm{g}^{\mathrm{ineq,relaxed}}(\lbar{\bm{X}},\lbar{\bm{U}}, \mathbf{K}, \bm{\Theta})],
\label{eq:NL-aug-cost}
\end{aligned} \end{align}
where $\bm{g}^{\mathrm{eq,relaxed}}$ and $\bm{g}^{\mathrm{ineq,relaxed}}$ encompass the constraints that are relaxed and appended to $P(\cdot)$ in \cref{prob:convexSOC}:
\begin{align} \begin{aligned}
&\bm{g}^{\mathrm{eq,relaxed}}(\lbar{\bm{X}},\lbar{\bm{U}}, \mathbf{K}, \bm{\Theta}) = \bar{\bm{x}}_0 + \sum_{k=0}^{N-1} \int_{t_k}^{t_{k+1}}\bm{f}(\bar{\bm{x}}, \lbar{\bm{u}}_k, t)\diff t - \bm{x}_f
,\\
&
\bm{g}^{\mathrm{ineq,relaxed}}(\lbar{\bm{X}},\lbar{\bm{U}}, \mathbf{K}, \bm{\Theta}) = \left\{g_{\mathrm{GA-impact}}(\lbar{\bm{x}}_{k_{\mathrm{GA}}}, P_{k_{\mathrm{GA}}}^{1/2}, \theta_{k_{\mathrm{GA}}}) \right\}_{\forall k_{\mathrm{GA}}} \quad \text{(\cref{eq:detCCperiapsis})}
\end{aligned} \end{align}
This formulation implies that we utilize \cref{eq:convexSOCcost} for $\Delta$V99 evaluation while penalizing the constraint violation of \cref{eq:terminalMeanOriginal,eq:stoCCperiapsis}, where we replace \cref{eq:terminalMeanOriginal} with the nonlinear deterministic propagation of the state under the mean control $\lbar{\bm{u}} $ and \cref{eq:stoCCperiapsis} with the deterministic yet non-convex impact chance constraint given by \cref{eq:detCCperiapsis}.
Note here that the other constraints in \cref{prob:originalSOC} are considered to be automatically satisfied as a result of solving \cref{prob:convexSOC}.

\subsubsection{Step Acceptance and Trust Region Update}

The algorithm then uses $\rho$ to determine acceptance of the current iteration and the update of $\Delta_{\mathrm{TR}} $.
In this work, we employ a different scheme than the one proposed in Ref.~\cite{Oguri2023b}.
Similar to \cite{Kumagai2025}, we employ a different scheme than the one proposed in Ref.~\cite{Oguri2023b}, as follows:
\begin{align}
\begin{cases}
\text{accept}& \text{if}\ \rho \in [1 - \eta_0, 1 + \eta_0]
\\
\text{reject}& \text{else}
\end{cases}
\label{eq:acceptance}
\end{align}
where $1 \geq \eta_0 > \eta_1 > \eta_2 > 0$.
The trust region is updated as follows:
\begin{align}
\Delta_{\mathrm{TR}} \gets
\begin{cases}
\min{\{\alpha_2 \Delta_{\mathrm{TR}}, \Delta_{\mathrm{TR}, \max}\}}	& \text{if}\ \rho \in [1 - \eta_2, 1 + \eta_2]
\\
\Delta_{\mathrm{TR}}	& \text{elseif}\ \rho \in [1 - \eta_1, 1 + \eta_1]
\\
\max{\{\Delta_{\mathrm{TR}} / \alpha_1, \Delta_{\mathrm{TR}, \min}\}} & \text{else}
\end{cases}
\label{eq:TRupdate}
\end{align}
where $\alpha_1>1$ and $\alpha_2>1$, and $0 < \Delta_{\mathrm{TR}, \min} < \Delta_{\mathrm{TR}, \max} $ are the lower and upper bounds of the trust region radius.

These algorithmic modifications of \texttt{SCvx*} are motivated by the fact that the convex subproblem \cref{prob:convexSOC} is based on \textit{inexact linearization} of the state covariance propagation.
See our recent work \cite{Kumagai2025} for further discussion on this aspect.

\subsection{Implementation}

\subsubsection{Sequential Convex Programming}

\cref{t:params} lists the parameters used for the presented SCP algorithm.
For trust region scaling, $D_x = D_u = D_{\theta} = 1$ are used.
In each iteration, we use CVX \cite{Grant2014} with MOSEK \cite{Mosek2017} in Matlab to solve \cref{prob:convexSOC}.
We normalize all the physical quantities by unit length $l_{\mathrm{unit}}$ and time $t_{\mathrm{unit}}$, which are defined as $l_{\mathrm{unit}} = 1$ AU and $t_{\mathrm{unit}} = \sqrt{\mu_{\mathrm{sun}} / l_{\mathrm{unit}}^3 } $.

\begin{table}[tb]
\centering
\caption{Parameters used for the SCP algorithm}
\label{t:params}
{\footnotesize
\begin{tabular}{lcccccccccccc}
\hline \hline
{Symbol} & $ \epsilon_{\mathrm{opt}} $ & $ \epsilon_{\mathrm{feas}} $ & $ \{\eta_0, \eta_1, \eta_2\} $ & $\{\alpha_1, \alpha_2, \beta, \gamma\} $ & $w^{(1)} $ & {$w_{\mathrm{max}} $} & $\Delta_{\mathrm{TR}}^{(1)} $ & $\Delta_{\mathrm{TR, min}} $  & $\Delta_{\mathrm{TR, max}} $ & $\tau$
\\ \hline
Value & $10^{-6}$ & $10^{-6}$ & $\{1.0,\, 0.5,\, 0.1\} $ & $\{2.0,\, 3.0,\, 2,\, 0.95 \}$ & $10^2$ & {$10^{10}$} & $0.1$ & $10^{-8}$ & $1.0$ & $1.1$
\\
\hline \hline
\end{tabular}}
\end{table}

\subsubsection{Safeguard}
As discussed in Ref.~\cite{Oguri2024c}, we sometimes encounter an undesirable situation where numerically ill-conditioned problems prevent a convex solver from finding the solution to \cref{prob:convexSOC}.
It is observed that such numerical issues are often due to large penalty weights, $w^{(i)}$.
Thus, as a safeguard, we apply the following remedy when the SCP algorithm encounters such numerical issues:
reduce $w^{(i)} $ by a factor $\beta > 1$, i.e., $w^{(i+1)} = w^{(i)} / \beta $, and re-run the convex programming without changing any other parameters or reference trajectory.

\subsubsection{Nonlinear Monte Carlo Simulation for Solution Verification}

After obtaining a solution, it is insightful to perform nonlinear Monte Carlo simulations to verify that the approximations made for the tractable formulation are reasonable.
To that end, we utilize extended Kalman filter (EKF) in place of KF given in \cref{eq:estProcessLin} for filtering, where the time update is replaced with the nonlinear mean propagation, the linear system matrices $A_k, B_k, \bm{c}_k, G_{\mathrm{exe},k}, G_k $ are propagated about the nonlinear mean state, and the innovation process utilizes the nonlinear model given in \cref{eq:innovationProcess}.
The Brownian motion in the SDE \cref{eq:SDE} is approximated in the same manner as in Eq. 64 of Ref.~\cite{Oguri2022c}.

We seek a way to convert the policy defined in \cref{eq:SOClinPolicy}, which takes the stochastic process $\bm{z}_k $ defined in \cref{eq:z-process} to calculate statistical TCMs, to a policy that is based on $\hat{\bm{x}}_k$, without affecting the optimization result.
Such conversion is desirable when applying optimized FPC policies in nonlinear Monte Carlo because it is less straightforward to calculate the nonlinear equivalent of the stochastic process $\bm{z}_k $.
Hence, we utilize the result stated in \cref{lemma:state-estimate-FPC} to equivalently convert \cref{eq:SOClinPolicy} (the $\bm{z}_k$-based FPC policy) to \cref{eq:FPC-stateEstimate}, which feeds the deviation of $\hat{\bm{x}}_k$ from $\lbar{\bm{x}}_k$ to calculate statistical TCMs.
With this $\hat{\bm{x}}_k$-based FPC policy, we can directly use the state estimate $\hat{\bm{x}}_k $ from EKF to calculate $\bm{u}_k $ in nonlinear Monte Carlo.

\begin{lemma}
\label{lemma:state-estimate-FPC}
The output feedback control policy $\bm{u}_k = \lbar{\bm{u}}_k + K_{k}\bm{z}_k$ given in \cref{eq:SOClinPolicy} is equivalently expressed as the following state-estimate history feedback control policy:
\begin{align}
\bm{u}_k = \bar{\bm{u}}_k + \sum_{i=0}^{k} \hat{K}_{k,i}(\hat{\bm{x}}_i - \bar{\bm{x}}_i)
\label{eq:FPC-stateEstimate}
\end{align}
where $\hat{K}_{k,i}$ correspond to sub-matrices of $\hat{\mathbf{K}}$ as follows:
\begin{align}
\hat{\mathbf{K}} = 
\begin{bmatrix}
\hat{K}_{0,0} & 0 & 0 & \cdots & 0 \\
\hat{K}_{1,0} & \hat{K}_{1,1} & \ddots & \cdots & 0 \\
\vdots & \vdots & \ddots & 0 & 0 \\
\hat{K}_{N-1,0} & \hat{K}_{N-1,1} & \cdots & \hat{K}_{N-1,N-1} & 0 \\
\end{bmatrix}
,\quad
\hat{\mathbf{K}} = \mathbf{K}(I + \mathbf{BK})^{-1}
\label{eq:FPC-stateEstimateGainDef}
\end{align}
where $\mathbf{K} $ solves \cref{prob:convexSOC}, and $(I + \mathbf{BK})$ is invertible.
\end{lemma}
\begin{proof}
See \cref{proof:state-estimate-FPC}.
\end{proof}

\begin{figure}[tb]
\centering 
\includegraphics[width=0.9\linewidth]{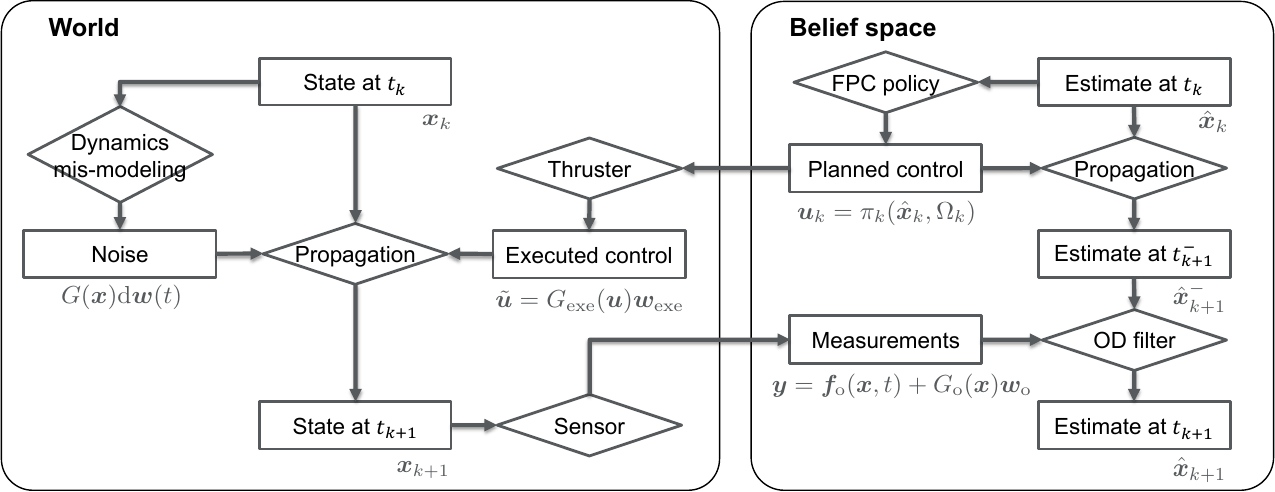}
	\caption{\label{f:MCflow} Nonlinear Monte Carlo flow chart, illustrating a simulation flow from $t_k$ to $t_{k+1}$.}
\end{figure}

\cref{f:MCflow} illustrates the flow of nonlinear MC simulations implemented in this study.
We propagate two sets of trajectories for each sample: true trajectory $\bm{x}(t) $ in the world and estimate trajectory $\hat{\bm{x}}(t) $ in the brief space.
For $\bm{x}(t) $, we incorporate noise acting on the acceleration due to dynamics mis-modeling (while approximating the Brownian motion in SDE in a similar manner to \cite{Oguri2022c}) in the nonlinear propagation as well as execution error based on Gates model.
In belief space, we calculate planned control $\bm{u}_k$ that incorporates statistical corrections based on the optimized FPC, propagate the estimate $\hat{\bm{x}}(t) $ nonlinearly, and perform OD to obtain $\hat{\bm{x}}_{k+1} $ via extended Kalman filter (EKF) by fusing measurements.
Again, calculation of $\bm{u}_k$ is based on the state-estimate history feedback FPC policy given by \cref{eq:FPC-stateEstimate};
it is observed that the state-estimate feedback FPC policy \cref{eq:FPC-stateEstimate} leads to MC trajectories that are more consistent with the linearly predicted covariance evolution, compared to the policy based on \cref{eq:SOClinPolicy}.
Although these two policies are equivalent in linear sense, their performance difference in nonlinear MC makes sense, since \cref{eq:SOClinPolicy} uses $\bm{z}$, which needs to be linearly propagated via \cref{eq:z-process}, whereas \cref{eq:FPC-stateEstimate} utilizes $\hat{\bm{x}} $, which may be estimated nonlinearly.

%!TEX root = main.tex
\section{Numerical Example}
\label{sec:examples}

\subsection{Scenario: Earth-Mars-Ceres transfer}

We consider an Earth-Mars-Ceres transfer with a GA at Mars (MGA) to demonstrate the proposed method.
The spacecraft is launched from Earth on 2030 Dec 18, performs the MGA on 2031 Aug 15, and arrives at Ceres on 2035 Aug 14, with the total time-of-flight (ToF) being $1700$ days ($\approx4.66$ years).
The position and velocity of the celestial bodies are defined based on the DE430 planetary ephemeris model provided by JPL.
The trajectory is discretized in time with intervals about $\Delta t_k = 45$ days, resulting in 7 nodes for the Earth-Mars leg and 34 nodes for the Mars-Ceres leg.
We assume $T_{\max} = 0.35 $ N, $m=3000$ kg, and $u_{\max} = T / m $, and the dynamics to be Keplerian about Sun, with instantaneous change in velocity at Mars due to the MGA.

The Earth launch is modeled as:
\begin{align}
\bm{r}_0 - \bm{r}_{\mathrm{Earth}}(t_0) = 0,\quad
\norm{\bm{v}_0 - \bm{v}_{\mathrm{Earth}}(t_0)}_2 \leq v_{\infty,\max},
\end{align}
that is, the spacecraft position is the same as that of Earth at the launch epoch $t_0$, while the velocity relative to Earth must be within the launcher's capability denoted by $v_{\infty,\max}$.
We assume $v_{\infty,\max} = 3.5$ km/s, which corresponds to $C_3 = v_{\infty,\max}^2 = 12.25\ \mathrm{km^2/s^2} $.

The MGA is modeled by using the patched-conic model (see \cref{sec:GA}).
We assume the Mars gravitational parameter to be $\mu_p = 42828\ \mathrm{km^3/s^2}$ and the minimum periapsis radius for Mars GA to be $r_{p\, \min} = 3689.5$ km, about 300 km altitude above the Mars surface.

The Mars arrival is modeled as rendezvous, i.e., the spacecraft position and velocity must match those of Mars at the arrival epoch.

\subsection{Deterministically-optimal Solution}
The transfer problem discussed above is first solved without considering uncertainties to generate the initial guess for the proposed method.
We solve the deterministic trajectory optimization problem using the SCP method presented in Ref.~\cite{Kumagai2024c} with $s_k = t_f - t_0 \ \forall k $, i.e., no time discretization mesh adaptation.

\cref{f:CeresDet} summarizes the deterministic optimal solution.
\cref{f:CeresDet_traj} shows the transfer trajectory projected onto the $x$-$y$ plane, where the origin is the Sun;
\cref{f:CeresDet_ctrl} shows the optimal control profile, where
$\alpha,\beta$ are the inertial longitude and latitude of the thrust direction, calculated as: $\alpha = \arctan2{(u_2/u_1)} $ and $\beta = \arcsin{(u_3)} $;
$\arctan2(\cdot)$ is the four-quadrant inverse tangent function.
$\norm{T} $ is the thrust magnitude in Newton.
Some statistics of the deterministic optimal solution are summarized in \cref{t:stats}, which indicates that the optimal trajectories fully leverages MGA (because $r_p = r_{p\, \min} $) and the launch capability ($v_{\infty,0} = \norm{\bm{v}_0 - \bm{v}_{\mathrm{Earth}}(t_0)}_2 = v_{\infty,\max} $), with $\Delta$V requirement 6.74 km/s;
RA and DEC in the table are the right ascension and declination of the launch in the inertial frame.
This table also includes the data for statistically-optimal solution, which is discussed in the next subsection.

\begin{figure}[tb]
\centering \subfigure[\label{f:CeresDet_traj} Deterministic optimal trajectory]
{\includegraphics[width=0.49\linewidth]{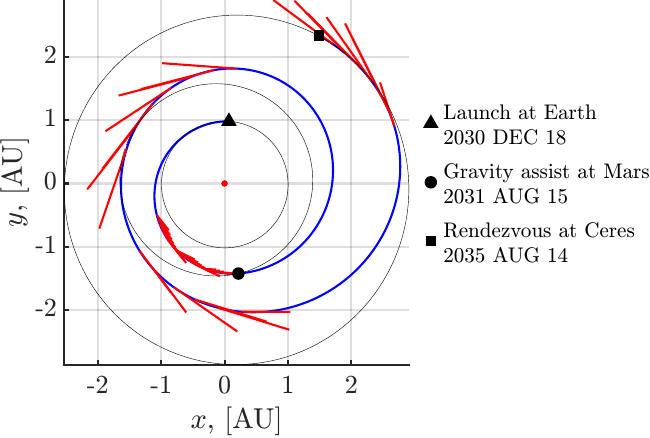}}
\centering \subfigure[\label{f:CeresDet_ctrl} Time history of deterministic optimal control]
{\includegraphics[width=0.49\linewidth]{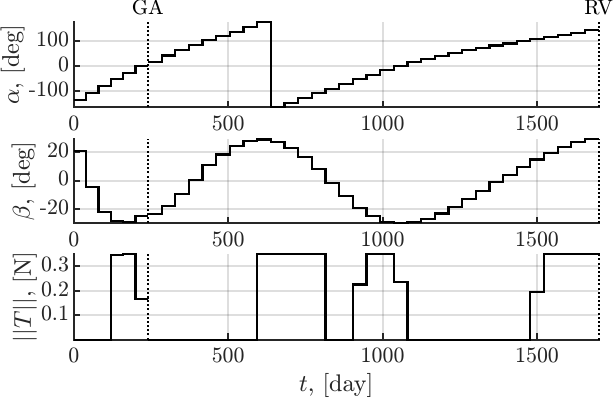}}
	\caption{\label{f:CeresDet} Deterministic optimal Earth-Mars-Ceres low-thrust transfer}
\end{figure}

\begin{table}[tb]
\centering
\caption{Comparison of deterministic-optimal and statistically-optimal solutions}
\label{t:stats}
\footnotesize
\begin{tabular}{lcccccc}
\hline \hline
& Launch: \{$v_{\infty,0}$, RA, DEC\} & MGA altitude: \{nominal, min\} & Total $\Delta$V: \{nominal, $\Delta$V99\}
\\ \hline
Deterministic
& \{3.5 km/s, -145.3 deg, 37.3 deg\} & \{300 km, N/A\} & \{6.74 km/s, N/A\}
\\ \hline
Statistical
& \{3.5 km/s, -144.7 deg, 36.0 deg\} & \{417 km, 368 km\} & \{6.81 km/s, 7.09 km/s\}
\\
\hline \hline
\end{tabular}
\end{table}

\subsection{Statistically-optimal Solution}

Now, we apply the proposed method to the same scenario while considering operational uncertainties and ensuring the robustness, including the statistical GA impact constraint and arrival dispersion constraint, minimizing $\Delta$V99.

\subsubsection{Operational Uncertainty Modeling}

We consider various operational uncertainties, including launch dispersion, OD error, maneuver execution error, and stochastic force.
These uncertainties are modeled as:
\begin{align}
P_0 = \mathtt{blkdiag}(\sigma_{r0}^2 I_3, \sigma_{v0}^2 I_3),\quad
G_{\mathrm{obs,phase}} = \mathtt{blkdiag}(\sigma_{r,\mathrm{phase}}^{\mathrm{OD}} I_3, \sigma_{v,\mathrm{phase}}^{\mathrm{OD}} I_3),\quad
G = \sigma_{\mathrm{acc}} \sqrt{\Delta t_{\mathrm{WN}}} \cdot I_3
\label{eq:uncertainty}
\end{align}
where $\mathtt{blkdiag}(\cdot) $ forms a block diagonal matrices.
\cref{t:uncertainty} summarizes specific values of the standard deviations of these uncertain quantities assumed for this numerical example.
$\Delta t_{\mathrm{WN}}$ denotes the time discretization step for approximation of SDE with white noise (see, for instance, Eq. 64 of \cite{Oguri2022c});
this roughly corresponds to experiencing zero-mean acceleration with standard deviation $\sigma_{\mathrm{acc}}$ with update interval $\Delta t_{\mathrm{WN}}$.
We take $\Delta t_{\mathrm{WN}} = 1$ hour.
Note that $(\sigma_{\mathrm{acc}} =) 1.0\ \mathrm{\mu m/s^2}$ is slightly greater than the solar radiation pressure acceleration Psyche spacecraft might experience at 1 AU (based on the spacecraft property in \cite{Oguri2020e}).
$\sigma_1 = \sigma_3 = 0$ is assumed in this example.

In \cref{eq:uncertainty}, the observation is modeled as a result of orbit determination processes over $\Delta t_k$ and hence $G_{\mathrm{obs}}$ takes the form of full-state measurements, where the OD solutions are delivered about every 45 days ($\approx \Delta t_k$).
However, the formulation given in \cref{sec:SOCformulation,sec:tractableFormulation} can accommodate any partial-state measurements, such as range, range-rate, Doppler shift, optical navigation (OpNav), among others;
in such cases, we can consider having some nodes where measurements are taken but not control is applied, to have more frequent measurements for realism.
For instance, \cite{Oguri2024} considers horizon-based OpNav measurements, where OpNav measurements are taken every 0.8 days while (impulsive) control is applied every 2.3 days.

Also note that ``phase'' in the subscript of $G_{\mathrm{obs}}$ corresponds to one of the phases: Launch, Cruise, MGA, and Arrival, to model different levels of OD uncertainties we may expect at different phases of the mission.
These values are summarized in \cref{t:nav}, where the OD error immediately after launch is greater than that of the cruise phase, while those immediately before MGA and Ceres arrival are smaller since we can expect more frequent orbit determination campaigns before these critical events;
also, OpNav information would likely become available as the spacecraft approaches Mars or Ceres.

\begin{table}[tb]
\centering
\caption{Operational uncertainty model. Values represent the standard deviations.}
\label{t:uncertainty}
\begin{tabular}{c|c|c|c|c|c|c|c}
\hline \hline
& \multicolumn{2}{c|}{Launch dispersion} & \multicolumn{2}{c|}{OD error (cruise)} & \multicolumn{2}{c|}{Execution error (Gates)} & \multirow[t]{2}{*}{Stoch. accel.} \\
\hline
& position & velocity & position & velocity & magnitude & pointing & \\
\hline
Symbol & $\sigma_{r0}$ & $\sigma_{v0}$ & $\sigma_{r,\mathrm{cruise}}^{\mathrm{OD}} $ & $\sigma_{v,\mathrm{cruise}}^{\mathrm{OD}} $ & $\sigma_{2}$ & $\sigma_{4}$ & $\sigma_{\mathrm{acc}}$
\\
\hline
Value & $1.0 \times 10^5$ km & 10 m/s & 30 km & 0.3 m/s & 0.5\% & 0.5 deg & 1.0 $\mathrm{\mu m/s^2} $ \\
\hline \hline
\end{tabular}
\end{table}

\begin{table}[tb]
\centering
\caption{OD error standard deviations at different phases (position and velocity)}
\label{t:nav}
\begin{tabular}{ccccccc}
\hline \hline
Cruise & Launch & MGA & Arrival
\\ \hline
$\{ \sigma_{r,\mathrm{cruise}}^{\mathrm{OD}},  \sigma_{v,\mathrm{cruise}}^{\mathrm{OD}}\} $ &
$\{3 \sigma_{r,\mathrm{cruise}}^{\mathrm{OD}}, 3 \sigma_{v,\mathrm{cruise}}^{\mathrm{OD}}\}$ &
$\{1/3 \cdot \sigma_{r,\mathrm{cruise}}^{\mathrm{OD}}, 1/3 \cdot \sigma_{v,\mathrm{cruise}}^{\mathrm{OD}}\}$ &
$\{0.1 \sigma_{r,\mathrm{cruise}}^{\mathrm{OD}}, 0.1 \sigma_{v,\mathrm{cruise}}^{\mathrm{OD}}\}$ 
\\
\hline \hline
\end{tabular}
\end{table}

\subsubsection{Statistical Constraints}

We impose chance constraints on the control magnitude and on the GA pariapsis radius and terminal distribution constraint.
The following parameters for these statistical constraints are assumed in the numerical example: 
\begin{align}
P_f = \mathtt{blkdiag}(\sigma_{rf}^2 I_3, \sigma_{vf}^2 I_3)
,\quad
\sigma_{rf} = 10^5\ \mathrm{km}
,\quad
\sigma_{vf} = 100\ \mathrm{m/s}
,\quad
\varepsilon_u = 
\varepsilon_{\mathrm{GA}} = 10^{-3}
\label{eq:constraint}
\end{align}

\subsubsection{Optimization Result}

The developed solution method is applied to solve the problem, which converged in 69 iterations of convex programming and took around 1.5 hours.
As pointed out in \cref{sec:intro}, a main bottleneck of the presented approach is the computational complexity, and some recent studies based on full covariance matrix \cite{Kumagai2024b,Kumagai2025} are found more computationally efficient at the cost of some minor drawbacks (see \cref{sec:intro} for more discussion).

\cref{f:CeresSto1} shows the optimization result, with
\cref{f:CeresStoTraj} illustrating the nominal trajectory and control directions projected on the x-y plane, while \cref{f:CeresStoCtrl} depicting the nominal control history (solid line) with the bottom plot overlaying the 99.9\% control magnitude bound under statistical FPC (dotted line).
First, it is evident that the control profile in \cref{f:CeresStoCtrl} is largely different from \cref{f:CeresDet_ctrl}.
In particular, we can see that the nominal thrust level is reduced at immediate before the MGA and arrival at Ceres, implying that margins are automatically introduced to the nominal control to accommodate necessary statistical FPC activities.
Conceptually, these margins are similar to heuristic margins quantity called \textit{forced coasting period} and \textit{duty cycle} \cite{Rayman2007}, which are introduced manually at sensitive locations on low-thrust trajectories.
The significance here is that this kind of margin quantities are discovered by the solution method with no specific inputs from mission designers, enabled by the stochastic optimal control approach that integrates OD and FPC processes into mission design.

\begin{figure}[tb]
\centering \subfigure[\label{f:CeresStoTraj} Trajectory x-y projection]
{\includegraphics[width=0.47\linewidth]{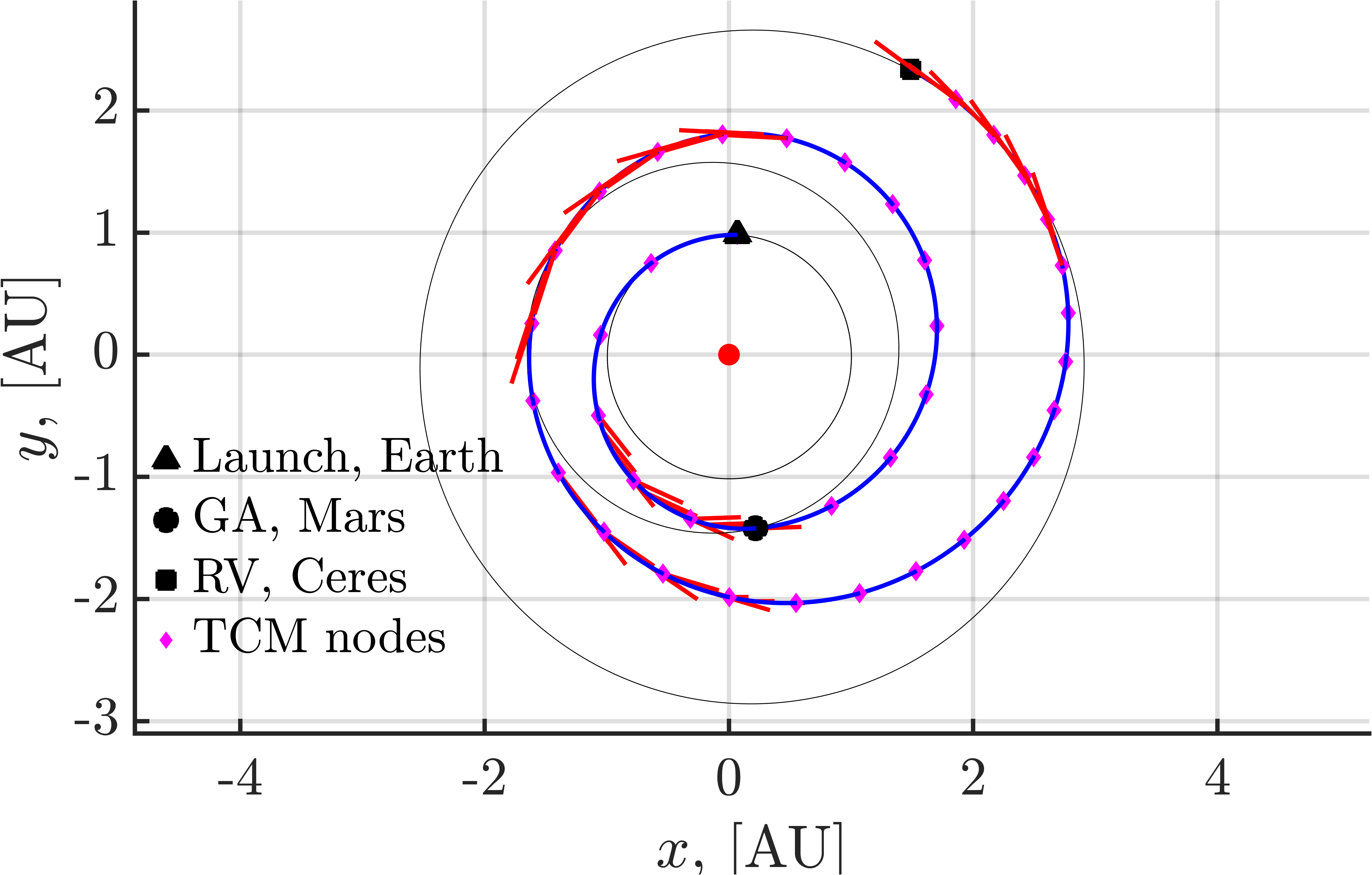}}
\centering \subfigure[\label{f:CeresStoCtrl} Control history: Nominal (solid) and 99.9\% bound (dotted)]
{\includegraphics[width=0.47\linewidth]{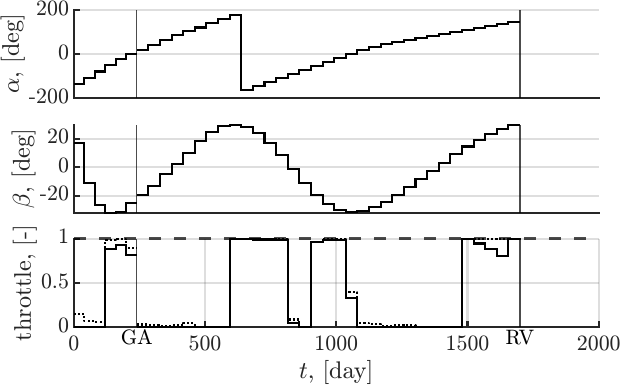}}
	\caption{\label{f:CeresSto1} Statistically-optimal Earth-Mars-Ceres low-thrust transfer solution}
\end{figure}

\subsubsection{Monte Carlo Simulation}
To confirm the robustness of the designed reference trajectory and associated FPC policies, nonlinear Monte Carlo (MC) simulations are performed with 100 samples.
Recall \cref{f:MCflow} for the flow of nonlinear MC simulations.

\begin{figure}[tb]
\centering \subfigure[\label{f:CeresStoTrajMC} MC trajectories projected on x-y plane.]
{\includegraphics[width=0.55\linewidth]{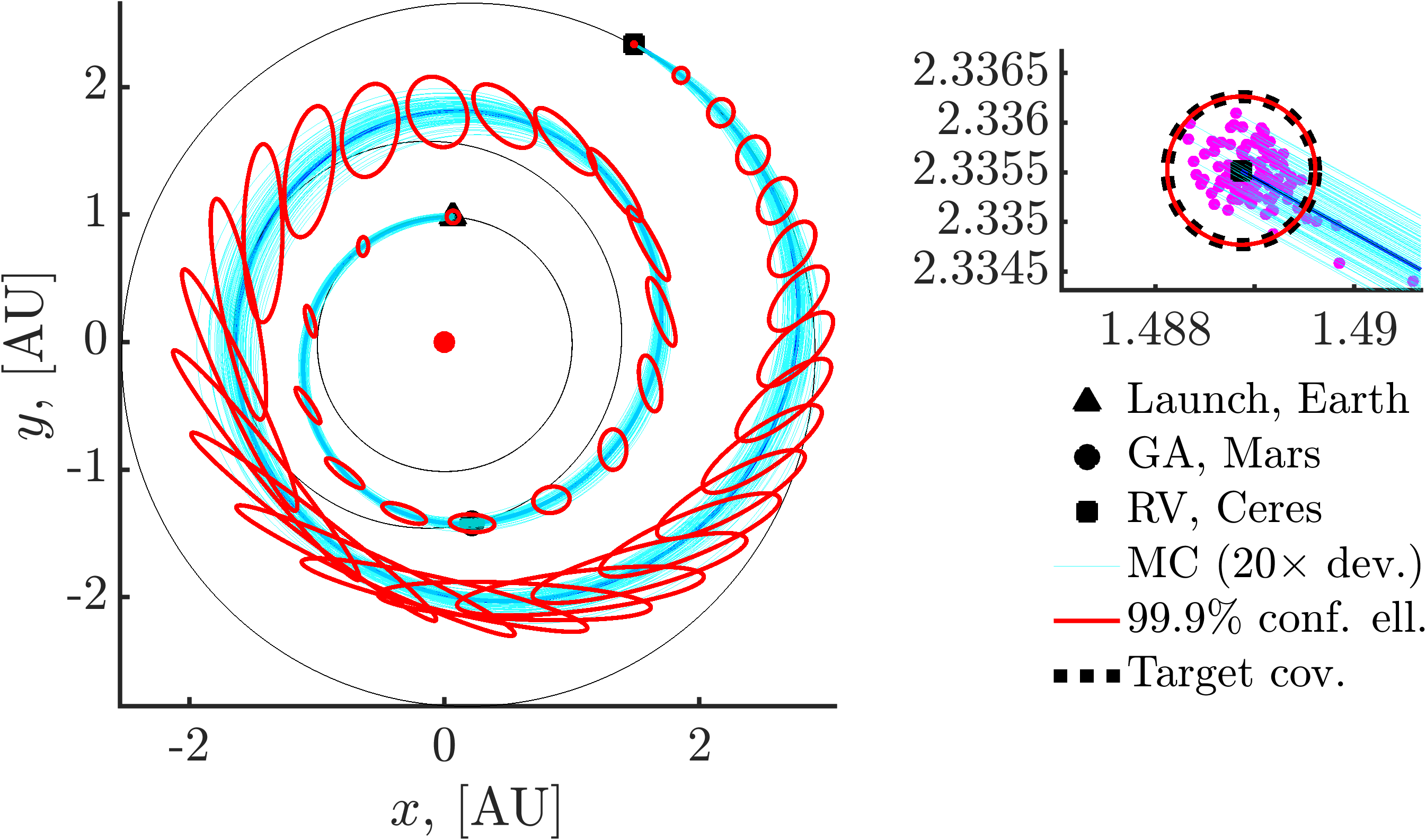}}
\centering \subfigure[\label{f:CeresStoTrajMC1} MC trajectories in three projections (not to scale).]
{\includegraphics[width=0.44\linewidth]{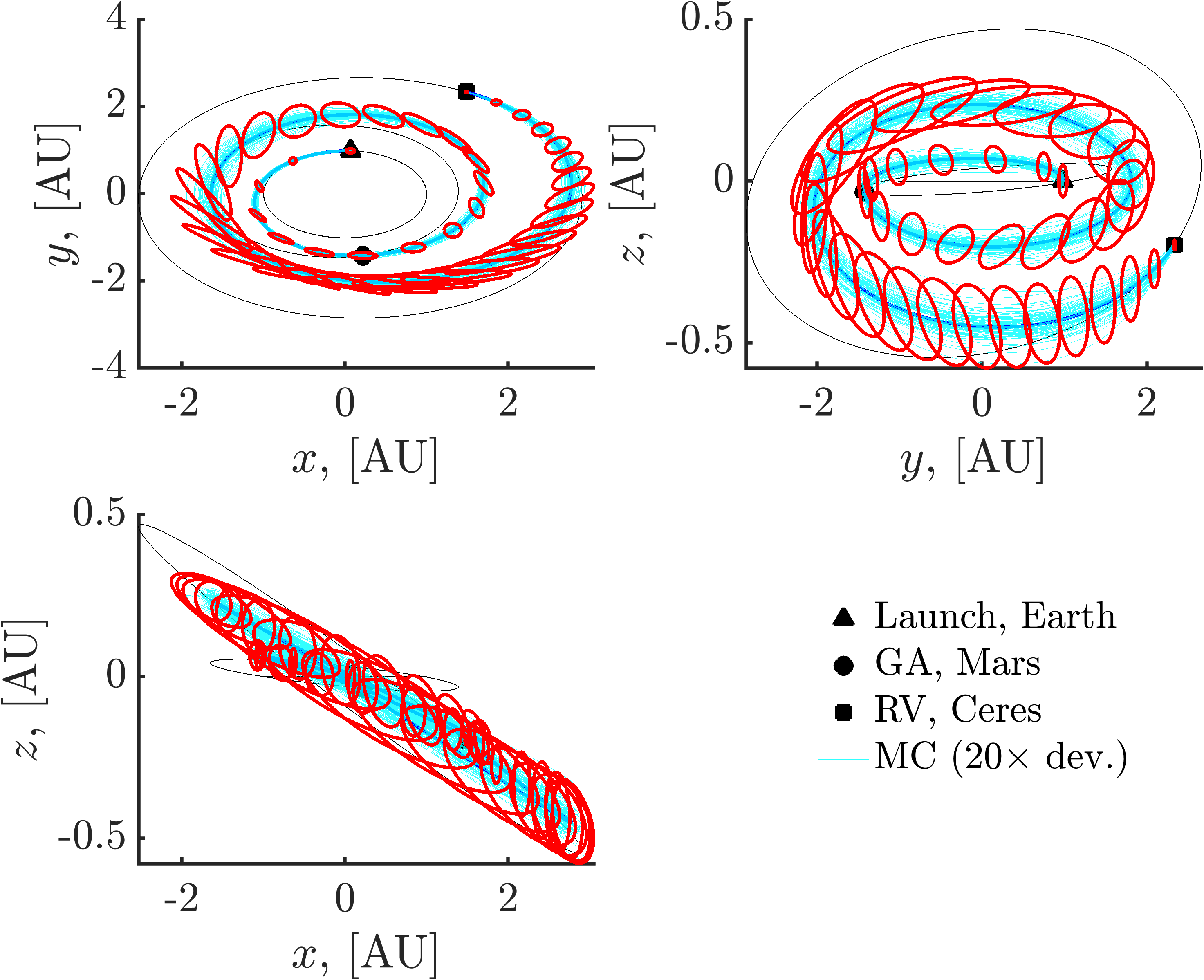}}
\centering \subfigure[\label{f:CeresStoTrajMCstate} State deviation: MC samples (gray) and 99.9\% confidence (black)]
{\includegraphics[width=0.53\linewidth]{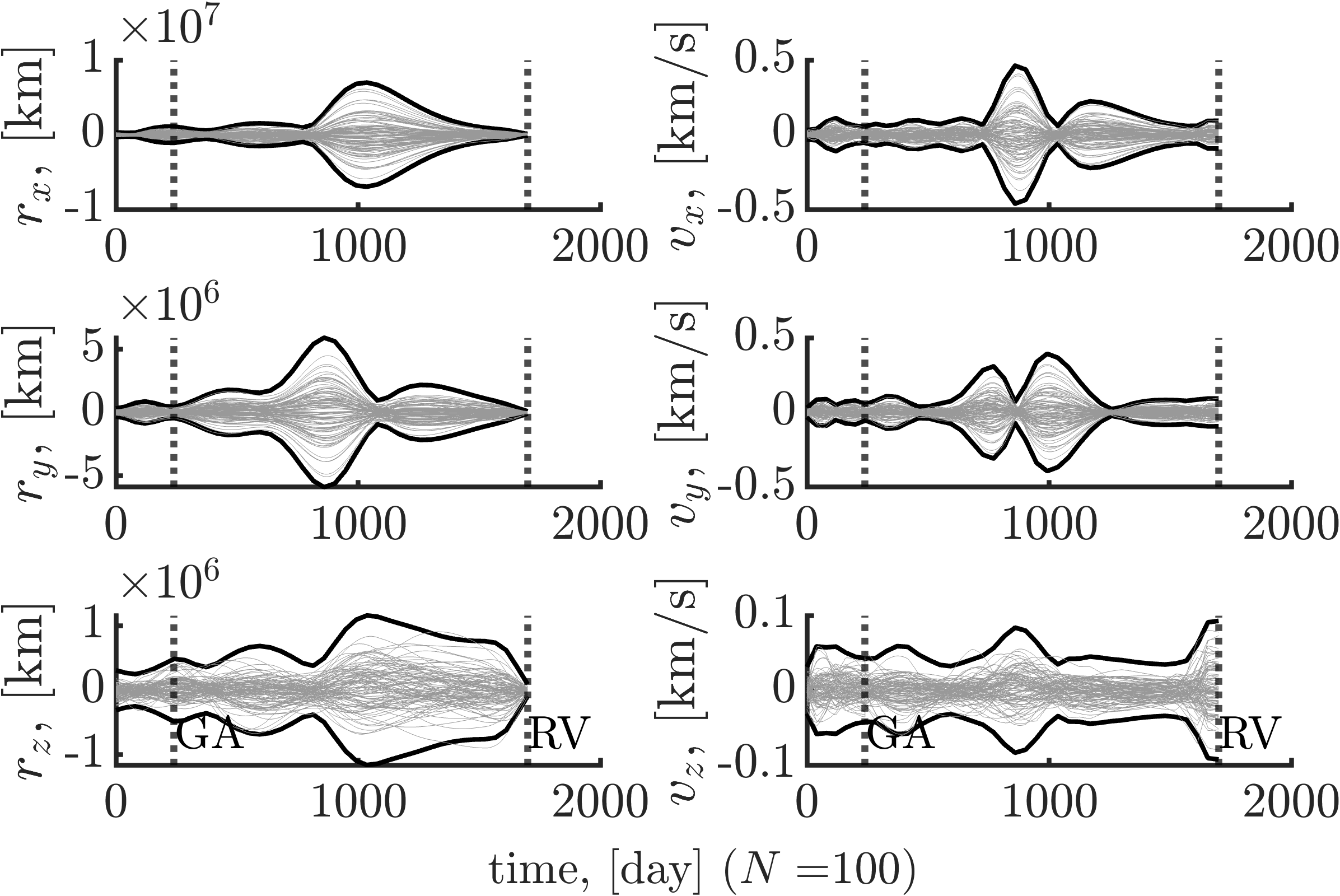}}
	\caption{\label{f:CeresStoMC} Monte Carlo simulation results. Curves in (a) and (b) represent: nominal (blue), MC samples (light blue), and 99.9\% confidence ellipses (red).}
\end{figure}

\cref{f:CeresStoMC} summarizes the MC trajectory results.
\cref{f:CeresStoTrajMC} shows the MC trajectories, projected on the x-y plane, where the deviation from the reference trajectory is amplified by 20 times for visualization purpose.
The red ellipses overlaid represent the covariance computed inside the proposed solution method (and used for calculation of $\Delta$V99 and statistical constraints), highlighting the overall accuracy of UQ results performed within the optimization.
The black ellipse in the zoomed view represents the covariance of the distributional constraint in \cref{eq:SOCterminal}, which confirms the successful arrival at Ceres with the prescribed accuracy, albeit some samples are not within the ellipse, indicating the effect of nonlinear dynamics.
\cref{f:CeresStoTrajMC1} includes the same trajectories but projected in x-y, y-z, and x-z planes, further demonstrating the accuracy of LinCov UQ in this case.
\cref{f:CeresStoTrajMCstate} shows the state component time profiles of the MC simulations, which illustrates:
1) state dispersion is kept relatively small until MGA, 
2) state dispersion significantly grows around the middle of the trajectory, which, as clear from \cref{f:CeresStoCtrl}, corresponds to the period where the optimized FPC does not apply much statistical corrections,
3) position dispersion rapidly decreases as the spacecraft approaches Ceres.

\begin{figure}[tb]
\centering \subfigure[\label{f:CeresStoCtrlMC} Control profile: nominal (black) and MC samples (blue)]
{\includegraphics[width=0.47\linewidth]{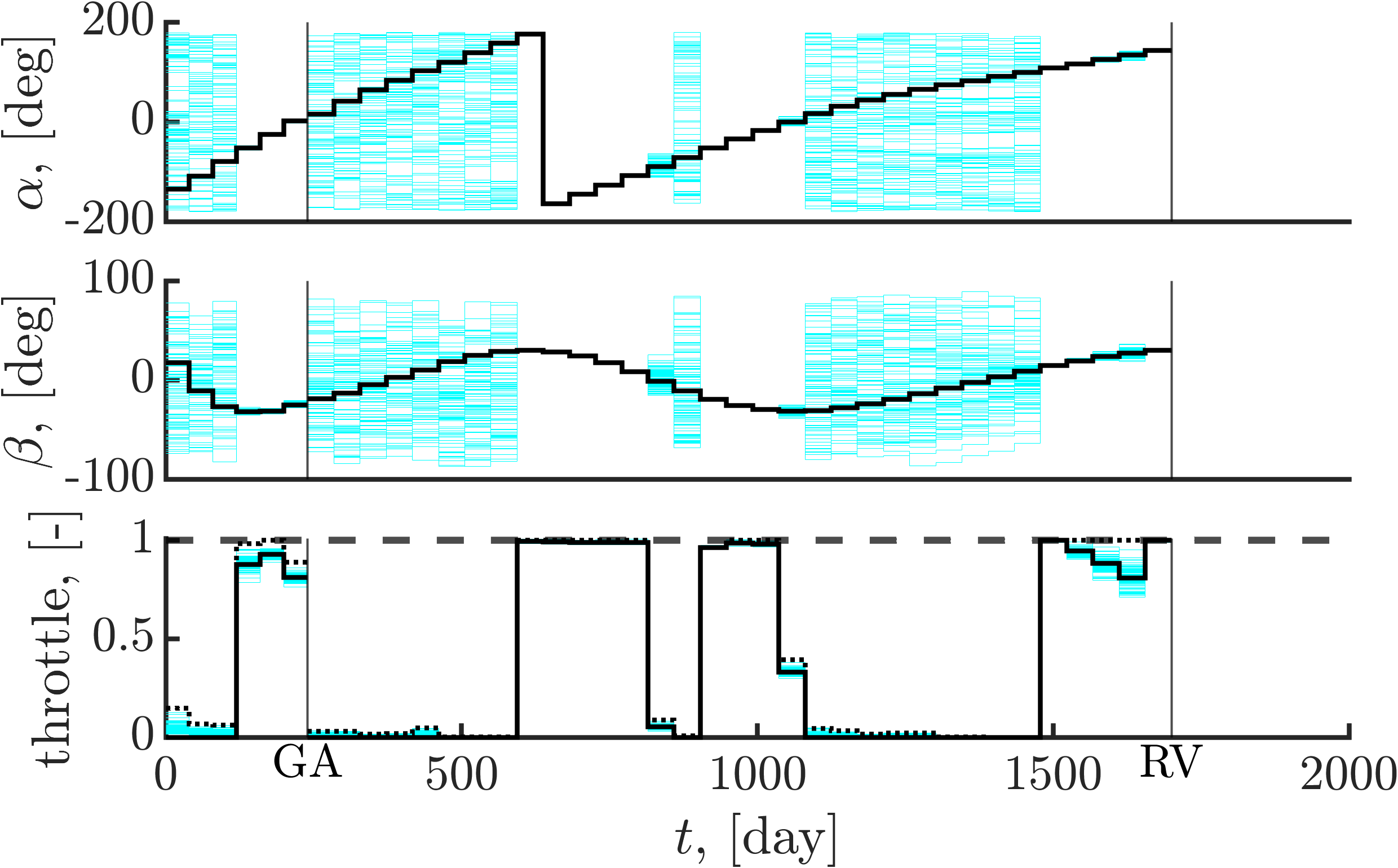}}
\centering \subfigure[\label{f:CeresStoNavMC} OD error via extended Kalman filter]
{\includegraphics[width=0.52\linewidth]{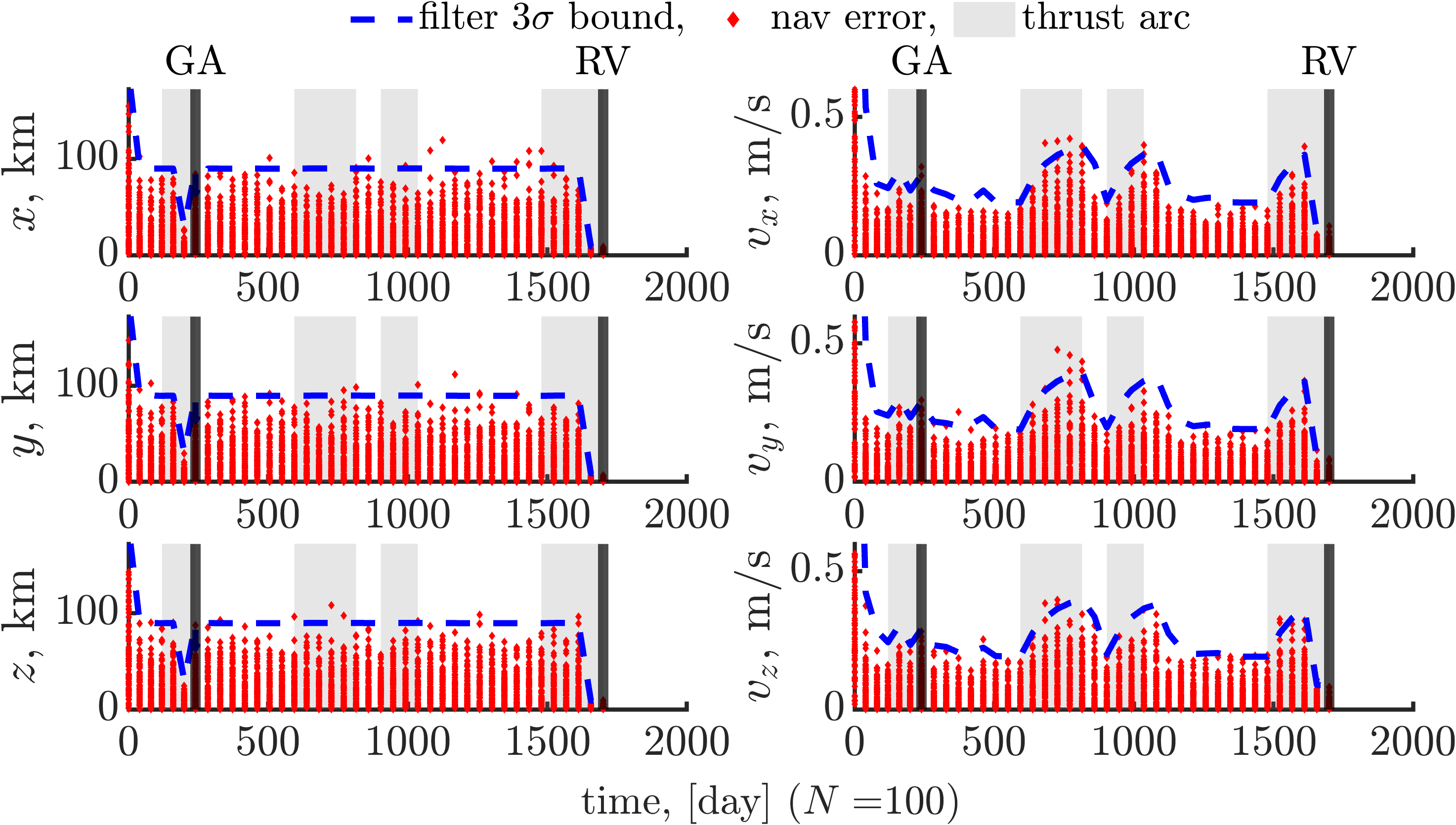}}
	\caption{\label{f:CeresStoData} Monte Carlo simulation results: Control profile and OD errors.}
\end{figure}

\cref{f:CeresStoData} display MC results in terms of the executed control profiles, which incorporates FPC and Gates execution error [\cref{f:CeresStoCtrlMC}], and the OD error time profiles based on EKF [\cref{f:CeresStoNavMC}].
\cref{f:CeresStoCtrlMC} clearly demonstrates that the maximum thrust constraint is satisfied under the action of statistical FPC, as expected from the imposed chance constraint.
\cref{f:CeresStoNavMC} shows the OD errors and three-sigma covariance bounds over time, with the thrusting arcs overlaid as gray regions, highlighting that all the OD solutions are successfully within the three-sigma bounds.
The change in the OD error levels reflects the different OD uncertainties defined in \cref{t:nav};
also, the velocity OD error is greater on thrusting arcs as expected.

\begin{figure}[tb]
\centering \subfigure[\label{f:CeresStoMGA} MGA periapsis distribution and statistics]
{\includegraphics[width=0.49\linewidth]{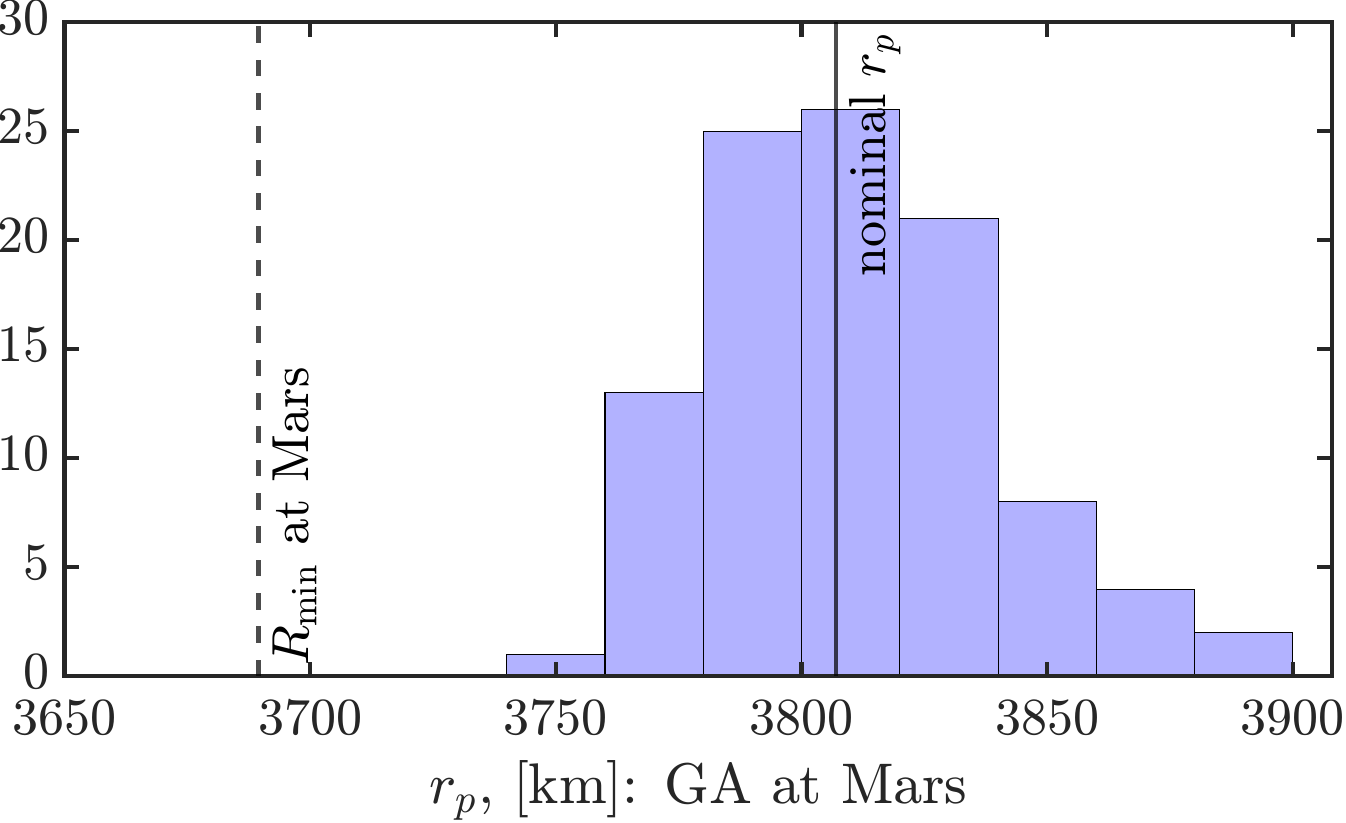}}
\centering \subfigure[\label{f:CeresStoDV} $\Delta$V distribution and statistics]
{\includegraphics[width=0.49\linewidth]{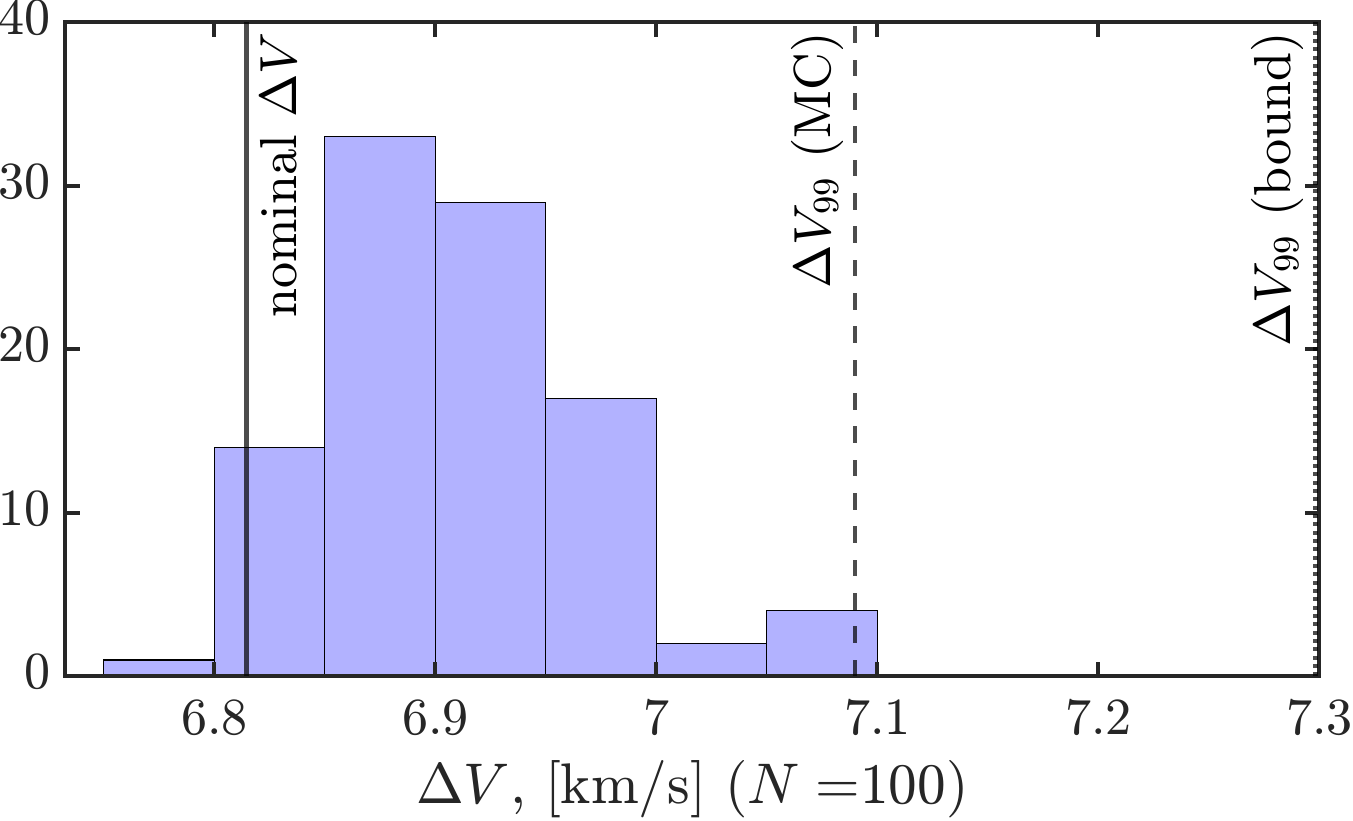}}
	\caption{\label{f:CeresStoStats} Monte Carlo simulation results: Statistics (y-axis: number of samples)}
\end{figure}

\cref{f:CeresStoStats} complies some important statistics of the MC results.
\cref{f:CeresStoMGA} shows the distribution of the periapsis radius at the Mars GA and demonstrates the satisfaction of the minimum MGA periapsis constraint under uncertainties thanks to the robust statistical FPC policies.
\cref{f:CeresStoDV} shows the distribution of the total $\Delta$V, with the nominal $\Delta$V, $\Delta$V99 calculated from MC samples, and the upper bound of $\Delta$V99, i.e., $J_{\mathrm{ub}} $ in \cref{eq:convexSOCcost}, which is minimized within the presented algorithm.
This figure shows that this bound is indeed upper-bounding the $\Delta$V99 cost, albeit with some conservativeness;
exploring different formulations that have smaller conservativeness is an interesting direction of future work.

\subsection{Discussion}

The numerical result demonstrates the effectiveness of the proposed method in designing statistically optimal, robust trajectories under state uncertainties due to launch dispersion, maneuver execution error, OD error, and imperfect dynamics modeling.
The designed trajectory satisfies critical constraints under uncertainties while minimizing the statistical cost, $\Delta$V99 (its upper bound, to be precise).
The formulation can naturally incorporate gravity assists (GAs), which underlines the flexibility of the proposed framework.

\cref{t:stats} compares the statistics of the deterministically and statistically optimal trajectories, which provides a few important observations.
First, the statistically optimal trajectory has a greater MGA altitude than the deterministic optimal counterpart, which represents a safe MGA margin automatically (and optimally) derived within the algorithmic framework of the proposed method.
This shift in the nominal MGA altitude naturally introduces a less effective MGA, which is compensated by the additional nominal $\Delta$V and largely altered thrust profiles (compare \cref{f:CeresDet_ctrl,f:CeresStoCtrl}).
This change in the nominal trajectory is also accompanied by the slight change in the optimal launch direction, about 0.5-1 degree change in the right ascension and declination, suggesting explicit incorporation of state uncertainty into mission design may affect the optimal launch condition.
In a similar vein, it is anticipated that the optimal epochs of critical events (e.g., launch, GA, arrival) would be also different for statistically and deterministically optimal trajectories;
incorporating time-varying events into the proposed SCP framework is straightforward by utilizing an adaptive-mesh SCP algorithm \cite{Kumagai2024c}.

There is one important caveat about the GA formulation presented in this paper.
Upon closer look at \cref{f:CeresStoTrajMCstate}, acute readers may wonder why the position dispersion at MGA is not reduced as much as we may intuitively expect;
one plausible reason for that is due to the simple (patched-conic) MGA modeling.
As it can be seen in \cref{eq:NLeqGA}, the effect of a GA is modeled as a rotation of the velocity relative to the body, where the rotation is parameterized by $\bm{u}_k$.
It is an important future work to consider a higher-fidelity GA model (e.g., treating a GA as a finite-duration phase centered around the body \cite{Ellison2019}) to capture the nonlinear effect more accurately, which may benefit from non-Gaussian UQ depending on the level of state uncertainty at the time of GA.

Higher-fidelity verification and validation would be also crucial when transitioning robust solutions generated by the proposed method to operation.
As discussed in Remark 1 of Ref.~\cite{Kumagai2025}, the linear FPC policy, whether \cref{eq:SOClinPolicy} or \cref{eq:FPC-stateEstimate}, can be interpreted as a sub-optimal substitute for optimization-based FPC policies.
Typical ground-based FPC for low-thrust missions re-optimizes the trajectory starting from the current state to the destination by solving the nonlinear trajectory optimization problem iteratively at a user-defined cadence (called \textit{design cycle});
for instance, an algorithm at JPL called \textit{Veil} implements this procedure \cite{Parcher2011,Lantoine2024}.
Thus, looking ahead, it will be of great importance to develop and apply Veil-like software or a simplified analysis procedure (e.g., replace the nonlinear re-optimization process in Veil with convex optimization; apply linear covariance steering to a deterministically optimal trajectory for comparison, etc).
Such analyses would further clarify benefits of the proposed approach and help bring this technology one step closer to operation with a higher technology readiness level (TRL).

Another important extension of this line of research is its application to solar-sail trajectory design.
Solar-sail missions are more likely to suffer from state uncertainties due to the unpredictable nature of the thrust forces caused by the interaction between the sunlight and sail surface, whose shape and optical properties are not easy to model and may be time-varying \cite{Dachwald2007}.
Ref.~\cite{Oguri2022d} is one of the early studies that applies a statistical approach to robust solar-sail trajectory design, which however did not demonstrate the full capability of statistical approaches, since it had to impose artificial constraints on the admissible feedback control due to the nonlinear reliance of solar-sail acceleration on the attitude variables.
An exciting avenue here is to combine a statistical approach with the recent work on a lossless control-convex formulation of solar-sail trajectory optimization problem \cite{Oguri2024c} to design statistically optimal, robust solar-sail trajectories under severe uncertainties.

%!TEX root = main.tex
\section{Conclusions}

This paper presents a mission design framework for statistically optimal trajectory design under uncertainty.
The developed framework leverages recent advances in stochastic optimal control and sequential convex programming to formulate the otherwise intractable problem into a sequence of convex problems.
The framework enables mission designers to concurrently design a reference trajectory and flight-path control plans such that minimize $\Delta$V99 while ensuring the satisfaction of mission constraints (e.g., planetary rendezvous with certain accuracy, no collision with a planet during a flyby, and more) under uncertainty due to launch dispersion, navigation (orbit determination + flight-path control) error, execution error, and stochastic disturbance.
The proposed framework is successfully applied to an Earth-Mars-Ceres transfer with an intermediate gravity assist at Mars.
Nonlinear Monte Carlo simulations demonstrate the robustness of the designed reference trajectory and flight-path control under the original, nonlinear stochastic dynamics and uncertain errors.

\section*{Acknowledgments}
The work described in this paper was carried out in part at the Jet Propulsion Laboratory, California Institute of Technology, under a contract with National Aeronautics and Space Administration.
K.Oguri's effort was supported in part by Purdue University.

\appendix
%!TEX root = main.tex
\section{Appendix}

\subsection{Proof of \cref{lemma:GAequation}: Nonlinear and Linear State Mapping under Gravity Assist} \label{proof:GAequation}
\begin{proof}
Deriving \cref{eq:NLeqGA} is straightforward by combining \cref{eq:GAcondition,eq:CayleyVinf} and noting that $[\bm{u}_k]_{\times} = V $.

To derive \cref{eq:LinEqGA}, let us begin with expressing \cref{eq:CayleyVinf} as
$\bm{v}_{k+1} - \bm{v}_p = (I_3 + [\bm{u}_k]_{\times})^{-1} (I_3 - [\bm{u}_k]_{\times})  (\bm{v}_k - \bm{v}_p)$,
which is equivalent to imposing $\bm{c}_{\mathrm{GA},v}(\bm{v}_{k+1}, \bm{v}_k, \bm{u}_k) = 0$, where
\begin{align} 
\bm{c}_{\mathrm{GA},v}(\bm{v}_{k+1}, \bm{v}_k, \bm{u}_k) = (I_3 + [\bm{u}_k]_{\times})(\bm{v}_{k+1} - \bm{v}_p) - (I_3 - [\bm{u}_k]_{\times}) (\bm{v}_k - \bm{v}_p)
\label{eq:GA-velConstraint}
\end{align}
Using the fact that $[\bm{u}_k]_{\times} \bm{a} = \bm{u}_k\times \bm{a} = - \bm{a}\times \bm{u}_k = - [{\bm{a}}]_{\times} \bm{u}_k $ for any vector $\bm{a}\in\R^3$, \cref{eq:GA-velConstraint} is equivalent to
\begin{align} 
\begin{aligned}
\bm{c}_{\mathrm{GA},v}(\bm{v}_{k+1}, \bm{v}_k, \bm{u}_k) 
&=
(\bm{v}_{k+1} - \bm{v}_k) - ([{\bm{v}}_{k+1}]_{\times} + [{\bm{v}}_{k}]_{\times} - 2[{\bm{v}}_p]_{\times})\bm{u}_k
\end{aligned}
\end{align}
Thus, its Taylor series expansion about $\bm{v}_{k+1}^*, \bm{v}_k^*, \bm{u}_k^*$ is given by [$(\cdot)^*$ indicates evaluation at the reference point]:
\begin{align} 
\bm{c}_{\mathrm{GA},v}(\bm{v}_{k+1}, \bm{v}_k, \bm{u}_k) &= 
\bm{c}_{\mathrm{GA},v}^* +
\left.\frac{\partial \bm{c}_{\mathrm{GA},v}}{\partial \bm{v}_{k+1}}\right|^* (\bm{v}_{k+1} - \bm{v}_{k+1}^*) + 
\left.\frac{\partial \bm{c}_{\mathrm{GA},v}}{\partial \bm{v}_{k}}\right|^* (\bm{v}_{k} - \bm{v}_{k}^*) + 
\left.\frac{\partial \bm{c}_{\mathrm{GA},v}}{\partial \bm{u}_{k}}\right|^* (\bm{u}_{k} - \bm{u}_{k}^*) + H.O.T.
\end{align}
where
\begin{align} 
\left.\frac{\partial \bm{c}_{\mathrm{GA},v}}{\partial \bm{v}_{k+1}}\right|^* = I_3 + [\bm{u}_k^*]_{\times}
,\quad
\left.\frac{\partial \bm{c}_{\mathrm{GA},v}}{\partial \bm{v}_{k}}\right|^* = - (I_3 - [\bm{u}_k^*]_{\times})
,\quad
\left.\frac{\partial \bm{c}_{\mathrm{GA},v}}{\partial \bm{u}_{k}}\right|^* = - ([{\bm{v}}_{k+1}^*]_{\times} + [{\bm{v}}_{k}^*]_{\times} - 2[{\bm{v}}_p]_{\times})
\end{align}
Rearranging the term, we obtain
\begin{align} 
\begin{aligned}
\bm{c}_{\mathrm{GA},v}(\bm{v}_{k+1}, \bm{v}_k, \bm{u}_k) 
=&
(I_3 + [\bm{u}_k^*]_{\times}) \bm{v}_{k+1} - (I_3 - [\bm{u}_k^*]_{\times}) \bm{v}_{k} - ([{\bm{v}}_{k+1}^*]_{\times} + [{\bm{v}}_{k}^*]_{\times} - 2[{\bm{v}}_p]_{\times}) \bm{u}_k + H.O.T.
\end{aligned}
\end{align}
Thus, to the first order, $\bm{c}_{\mathrm{GA},v}(\bm{v}_{k+1}, \bm{v}_k, \bm{u}_k) = 0$ is equivalent to
\begin{align} 
\begin{aligned}
(I_3 + [\bm{u}_k^*]_{\times}) \bm{v}_{k+1} &= (I_3 - [\bm{u}_k^*]_{\times}) \bm{v}_{k} + ([{\bm{v}}_{k+1}^*]_{\times} + [{\bm{v}}_{k}^*]_{\times} - 2[{\bm{v}}_p]_{\times}) \bm{u}_k 
\ \Leftrightarrow\ 
\\
\bm{v}_{k+1} &= (I_3 + [\bm{u}_k^*]_{\times})^{-1} (I_3 - [\bm{u}_k^*]_{\times}) \bm{v}_{k} + (I_3 + [\bm{u}_k^*]_{\times})^{-1}([{\bm{v}}_{k+1}^*]_{\times} + [{\bm{v}}_{k}^*]_{\times} - 2[{\bm{v}}_p]_{\times}) \bm{u}_k 
\end{aligned}
\end{align}
which, together with $\bm{r}_{k+1} = \bm{r}_k $, leads to $A_k, B_k, \bm{c}_k $ given above.
\end{proof}

\subsection{Proof of \cref{lemma:state-estimate-FPC}: State-estimate Feedback Control Policy Conversion} \label{proof:state-estimate-FPC}
\begin{proof}
We show that the two control policies \cref{eq:SOClinPolicy} and \cref{eq:FPC-stateEstimate} give the same filtered state trajectory.

First, recall from \cref{eq:FPC-stateEstimateGainDef} that $\hat{\mathbf{K}} = \mathbf{K}(I + \mathbf{BK})^{-1}$, where $(I + \mathbf{BK})$ is invertible since $\mathbf{B}$ is strictly block lower-triangular (see \cref{eq:blockMatrixWrittenOut}) and $\mathbf{K}$ is block lower-triangular (see Ref.~\cite{Oguri2024}).
Using Eq.~41 to 43 of Ref.~\cite{Ridderhof2020} implies that $\mathbf{K}$ and $\hat{\mathbf{K}}$ satisfy
\begin{align}
(I - \mathbf{B\hat{K}})^{-1}  = (I + \mathbf{BK})
\label{eq:equivalency-Khat-and-K}
\end{align}
where note that $\mathbf{K} $ and $\hat{\mathbf{K}} $ here correspond to $F$ and $K$ in Ref.~\cite{Ridderhof2020}, respectively.
Also, since $\mathbf{B}$ is strictly block lower-triangular and $\mathbf{K}$ is block lower-triangular (see \cref{eq:FPC-stateEstimateGainDef}), $(I - \mathbf{B\hat{K}})$ is also invertible.

Then, under the policy \cref{eq:SOClinPolicy}, the filtered state trajectory is given as (see Eq.~45 of Ref.~\cite{Oguri2024}):
\begin{align}\begin{aligned}
\hat{\bm{X}} - \lbar{\bm{X}} 
&=
(I + \mathbf{B}\mathbf{K}) [\mathbf{A}(\hat{\bm{x}}_0^- - \lbar{\bm{x}}_0) + \mathbf{L}{\bm{Y}}]
\label{eq:blockStateDeviation}
\end{aligned}\end{align}
On the other hand, under the policy \cref{eq:FPC-stateEstimate}, $(\hat{\bm{X}} - \lbar{\bm{X}})$ is given as (see Eq.~40 of Ref.~\cite{Ridderhof2020}; notation is adapted for consistency):
\begin{align}\begin{aligned}
\hat{\bm{X}} - \lbar{\bm{X}} 
&=
(I - \mathbf{B\hat{K}})^{-1} [\mathbf{A}(\hat{\bm{x}}_0^- - \lbar{\bm{x}}_0) + \mathbf{L}{\bm{Y}}]
\label{eq:blockStateDeviation2}
\end{aligned}\end{align}
\cref{eq:equivalency-Khat-and-K} implies that \cref{eq:blockStateDeviation,eq:blockStateDeviation2} represent the same filtered state trajectory, completing the proof.
\end{proof}

% \bibliographystyle{unsrt}
% \bibliography{../../../../../../../utility/latex/ref/zotero/library}

\end{document}